\documentclass[sigconf]{acmart}
\AtBeginDocument{%
  }
\usepackage[utf8]{inputenc}
\DeclareUnicodeCharacter{202F}{\,}
\DeclareUnicodeCharacter{2011}{-}
\usepackage{pifont} 
\usepackage{subcaption}
\usepackage{float}

\acmConference[ICCPS '26]{International Conference on Cyber-Physical Systems}{May 11-14, 2026}{Saint Malo, France}

\usepackage{pifont} 
\usepackage{amsthm} 
\usepackage{amsmath} 
\newtheorem{theorem}{Theorem}

\usepackage{mathtools}

\newtheorem{definition}{Definition}

\usepackage{array}  
\usepackage{rotating}
\usepackage{makecell}

\usepackage{subcaption}
\usepackage[linesnumbered,ruled,vlined]{algorithm2e}
\SetKwInput{KwData}{Input}
\SetKwInput{KwResult}{Output}

\usepackage{mathtools}

\copyrightyear{2026}
\acmYear{2026}
\setcopyright{rightsretained} 
\acmConference[ICCPS '26]{ACM/IEEE International Conference on Cyber-Physical Systems}{May 11--14, 2026}{St. Malo, France}
\acmBooktitle{Proceedings of the 17th ACM/IEEE International Conference on Cyber-Physical Systems (ICCPS), Saint Mao, France, May 11-14, 2026}

\begin{document}

\title{ReACT-TTC: Capacity-Aware Top Trading Cycles for Post-Choice Reassignment in Shared CPS}
\author{Anurag Satpathy}
\authornote{*The authors contributed equally to this work.}
\email{anurag.satpathy@mst.edu}
\affiliation{%
  \department{Department of Computer Science}
  \institution{Missouri University of Science and Technology}
  \city{Rolla}
  \state{MO}
  \country{USA}
}

\author{Arindam Khanda}
\authornotemark[1]
\email{akkcm@mst.edu}
\affiliation{%
  \department{Department of Computer Science}
  \institution{Missouri University of Science and Technology}
  \city{Rolla}
  \state{MO}
  \country{USA}
}

\author{Chittaranjan Swain}
\email{cswain@iiitm.ac.in}
\affiliation{%
  \institution{Indian Institute of Information Technology and Management Gwalior}
  \city{Gwalior}
  \state{Madhya Pradesh}
  \country{India}
}

\author{Sajal K. Das}
\email{sdas@mst.edu}
\affiliation{%
  \department{Department of Computer Science}
  \institution{Missouri University of Science and Technology}
  \city{Rolla}
  \state{MO}
  \country{USA}
}

\thanks{This is a preprint of a paper accepted in 17th ACM/IEEE International Conference on Cyber-Physical Systems (ICCPS), Saint Mao, France, May 11-14, 2026.}


\begin{abstract}
Cyber-physical systems (CPS) increasingly manage shared physical resources in the presence of human decision-making, where system-assigned actions must be executed by users or agents in the physical world. A fundamental challenge in such settings is user non-compliance: individuals may deviate from assigned resources due to personal preferences or local information, degrading system efficiency and requiring light-weight reassignment schemes. This paper proposes a post-deviation reassignment framework for shared-resource CPS that operates on top of any initial allocation algorithm and is invoked only when users diverge from prescribed assignments. We advance the Top-Trading-Cycle (TTC) mechanism to enable voluntary, preference-driven exchanges after deviation events, and extend it to handle many-to-one resource capacities and unassigned resource conditions that are not supported by the classical TTC. We formalize these structural cases, introduce capacity-aware cycle-detection rules, and prove termination along with the preservation of Pareto efficiency, individual rationality, and strategy-proofness. A Prospect-Theoretic (PT) preference model is further incorporated to capture realistic user satisfaction behavior. We demonstrate the applicability of this framework on an electric-vehicle (EV) charging case study using real-world data, where it increases user satisfaction and effective assignment quality under non-compliant behavior.
\end{abstract}

\maketitle

\section{Introduction} \label{sec:intro}
Cyber-Physical Systems (CPS) increasingly manage shared physical resources in large-scale, interactive environments such as mobility services~\cite{chen2025mobility, nezami2025computing}, electric vehicle charging~\cite{khanda2025smevca, chatterjee2025v2vdiscs}, smart grids~\cite{cintuglu2016survey}, shared micromobility~\cite{liu2025spatial}, autonomous parking~\cite{chen2022parallel}, and urban logistics~\cite{khanda2025cargo}. In these systems, a digital controller computes resource allocations or schedules, while physical agents (e.g., vehicles, drones) execute decisions in the environment. In this context, effective operation requires efficient allocation of limited resources, rapid adaptation to changing conditions, and mechanisms that account for human behavior and heterogeneous user preferences.

A key challenge in CPS environments is user \textit{non-compliance}~\cite{stephens1998don}. Even when a controller assigns routes, schedules, or resources, users often deviate based on personal convenience or local observations such as queue length, travel distance, or price~\cite{hu2019modeling, chung2020intelligent}. For example, in the EV charging scenario in Fig.~\ref{fig:system_model},  EV $e_2$ was recommended to charge at $C_3$, but the user ignores it, prompting a real-time reassignment to $C_4$. Such behavior is widespread: 30--60\% of drivers override navigation recommendations~\cite{kerkman2012car, djavadian2014empirical, yu2022modeling}; 20--40\% of ride-hailing users reject assigned trips~\cite{ashkrof2025implications, tu2024effect}; shared micromobility users switch parking spots when docks appear full~\cite{hanowski1998effects}; 15--35\% of drivers ignore parking guidance when walking time increases~\cite{ji2014understanding}; 25--50\% of EV drivers divert from recommended chargers due to pricing or queues~\cite{pan2019modeling, alinia2019online}; 20--40\% of households disregard load-shifting instructions in demand-response programs~\cite{white2024time}; and logistics drivers routinely re-sequence deliveries~\cite{hou2022optimization}. These deviations violate allocation assumptions, reduce system-wide efficiency, and introduce compounding congestion or imbalance, underscoring the need for CPS mechanisms that explicitly model human decision-making and support real-time reassignment, and maintain user satisfaction.

\begin{figure}[h]
    \centering
\includegraphics[width=\columnwidth]{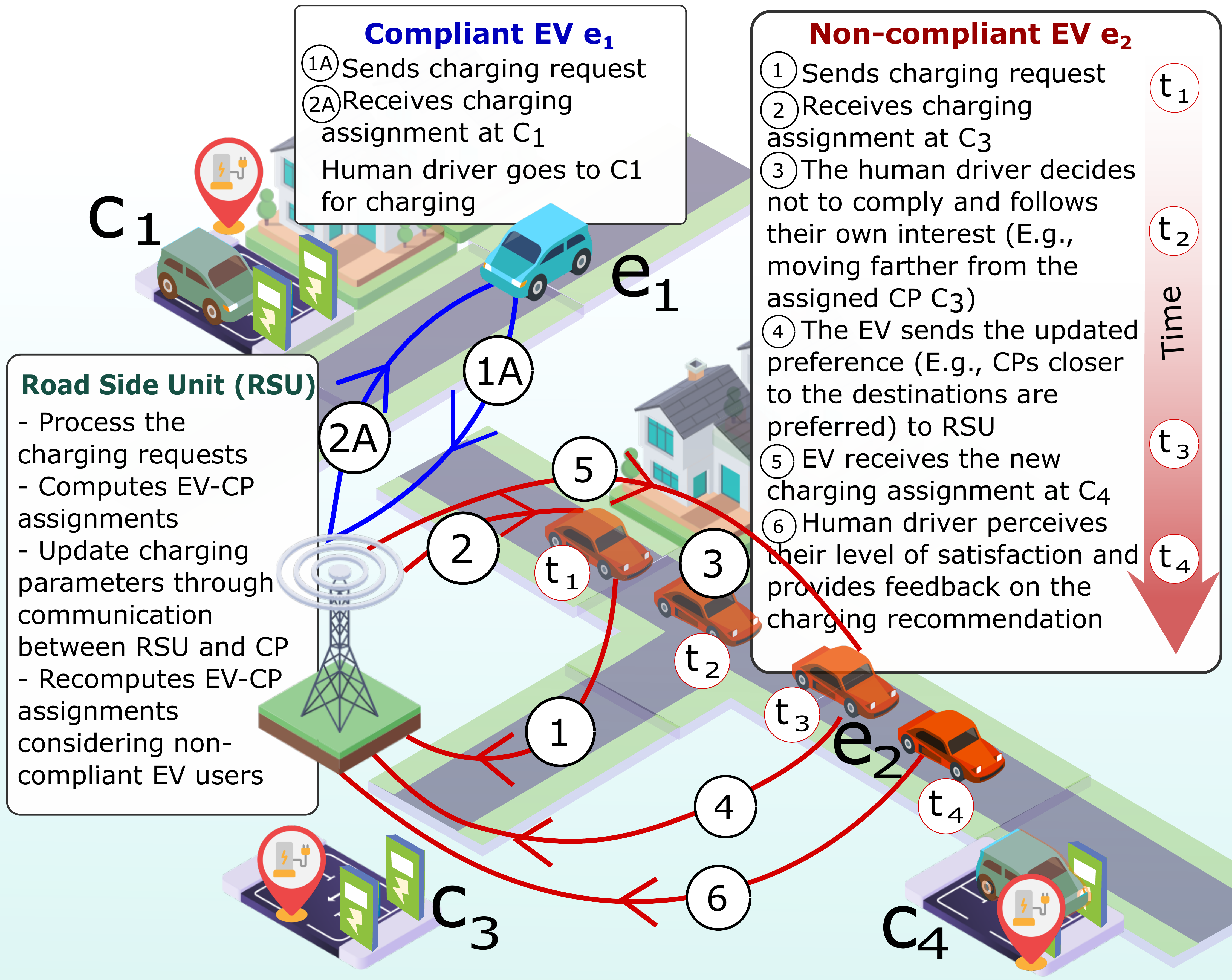}
    \vspace{-0.2in}
    \caption{An electric vehicle (EV) to charging point (CP) assignment scenario with compliance and non-compliance.}
    \vspace{-0.1in}
    \label{fig:system_model}
\end{figure}

Existing CPS coordination frameworks, including market-based scheduling~\cite{parhizi2016market}, heuristic dispatch~\cite{irfan2025optimizing}, and game-theoretic matching~\cite{akter2024mime}, largely assume that agents comply with assigned resources or routes. When deviations occur, most systems either ignore them or rely on centralized re-optimization, which is computationally expensive and disrupts already-committed allocations. Moreover, classical exchange-based mechanisms, such as classical Top-trading cycle (TTC)~\cite{morrill2024top} and barter-style assignment~\cite{qian2025smart}, assume one-to-one ownership and do not naturally extend to CPS settings where (i) resources have capacity, (ii) multiple agents may share a resource, and (iii) unused capacity may exist. These structural differences preclude direct use of exchange mechanisms and constrain real-time, user-aware adaptation to non-compliance.

To address these limitations, we propose a post-deviation reassignment layer that operates above existing CPS allocation mechanisms. When users deviate from assigned resources, our framework enables voluntary reassignment via a modified Top Trading Cycle (TTC) process that handles capacity-constrained and partially allocated resources, maintaining/improving user satisfaction~\cite{gale1974trade}. Unlike classical TTC, which assumes one-to-one ownership, CPS settings involve shared resources (co-ownerships), quotas, and idle capacity (no-ownership). We formally characterize the structural cases arising from capacity-constrained CPS resources (shared, under-filled, and idle slots), develop resource trading procedures, and prove that the adapted TTC mechanism preserves Pareto efficiency, individual rationality, and strategy-proofness in each case. 
Although broadly applicable to shared-resource CPS settings, we demonstrate its effectiveness in EV charging as a representative case. This work makes the following key \textbf{contributions}:

\begin{itemize}
    \item We propose a post-deviation reassignment layer for shared-resource CPS that augments any base allocation algorithm and activates only when users deviate.
    
    \item We develop a capacity-aware extension of TTC to enable voluntary, preference-driven exchanges in many-to-one settings with unassigned capacity.
    
    \item We identify structural cases in CPS resources (multi-user allocations and idle capacity) and derive resource-exchange rules to ensure the validity of trades and maximize user satisfaction.
    
    \item We prove termination and show that Pareto efficiency, individual rationality, and strategy-proofness are preserved under different cases.
    
    \item We incorporate a Prospect-Theoretic satisfaction model to more accurately capture human preference behavior compared to linear scoring.
    
    \item We evaluate the framework on real-world EV charging traces, demonstrating at least $43\%$ improvement in user satisfaction and efficient reassignment under non-compliance with minimal computation overhead.
\end{itemize}

The remainder of the paper is organized as follows. Section~\ref{sec:survey} reviews related work on non-compliant CPS systems. Section~\ref{sec:model_new} presents the system model and problem formulation, while Section~\ref{sec:solution_Approach_new} details the assignment problem and proposed solution. Theoretical analysis is provided in Section~\ref{sec:theory}, followed by experimental validation in Section~\ref{sec:results_new}. Section~\ref{sec:discussion} discusses limitations and future directions, and Section~\ref{sec:cnls_new} concludes the paper.

\vspace{-0.1in}
\section{Literature Review}\label{sec:survey}
Prior work spans four areas: non-compliance in CPS, dynamic resource allocation, market design mechanisms, and EV charging systems, collectively addressing bounded rationality, centralized compliance, and shared resource coordination.

\vspace{-0.1in}
\subsection{Non-compliant Behavior in CPS}
Studies on route-choice behavior show that compliance depends not only on system accuracy but also on how guidance is framed. In \cite{kerkman2012car}, drivers followed recommendations framed for collective efficiency but ignored those emphasizing personal optimization or providing excessive detail. Similarly, social-navigation experiments in \cite{djavadian2014empirical} found that transparency and well-designed incentives improve adherence, while poor framing or low trust induces deviation.

Similar patterns appear across CPS domains. In ride-sourcing systems, studies~\cite{ashkrof2025implications,tu2024effect} show that drivers’ trip acceptance depends on profitability, pickup distance, and surge pricing, which in turn affect wait times and throughput. In parking systems~\cite{ji2014understanding}, drivers often ignore digital guidance due to outdated data or habitual preferences. In crowd-sourced delivery~\cite{hou2022optimization}, couriers reject orders when incentives misalign with perceived effort. In energy and EV-charging coordination, scheduling frameworks remain strategy-proof only under truthful participation and fail to reassign resources once users deviate~\cite{alinia2019online}. Collectively, these studies highlight that effective CPS coordination must adapt to, rather than enforce, user behavior, motivating fast post-deviation reassignment.

\vspace{-0.1in}
\subsection{Dynamic Resource Allocation in CPS}
Cooperative and optimization frameworks have significantly advanced resource allocation in UAV-enabled CPS. Game-theoretic models achieve Pareto-efficient sharing of communication and control resources via Nash bargaining and cooperative task assignment~\cite{zhao2020allocation}, while multi-UAV systems employ energy-aware clustering and Shapley value–based fairness for balanced exploration and workload~\cite{guang2025game}. Complementary optimization schemes integrate goal-oriented bandwidth control and hybrid FDMA–NOMA trajectory planning to jointly minimize energy consumption and latency in aerial sensing networks~\cite{liu2025energy}. 

Recent studies extend CPS coordination to energy infrastructures, emphasizing adaptive control, distributed optimization, and resilience. Reinforcement learning–based digital twin frameworks improve grid efficiency under dynamic conditions~\cite{xiong2025adaptive}, while the integrated edge–fog architectures enable secure, low-latency, and energy-aware decisions with renewable-powered AI sensing~\cite{hu2025integrating}. Inverter-based microgrids use event-triggered distributed control for voltage–frequency stability under communication delays~\cite{wu2021distributed}. 

EV charging systems exemplify the challenges of large-scale CPS coordination, where real-time scheduling must balance user preferences and infrastructure constraints. Multi-objective models optimize charger selection by distance, state-of-charge (SoC), and capacity~\cite{algafri2024smart}, while distributed coordination schemes balance grid load through localized decisions~\cite{liu2022collaborative}. Transactive formulations integrate renewables, pricing, and user preferences via robust optimization~\cite{kabiri2024transactive}. Recent V2V charge-sharing frameworks model donor–acceptor pairing as a stable matching problem under communication limits, enabling decentralized energy exchange~\cite{chatterjee2025v2vdiscs}.

\vspace{-0.1in}
\subsection{Exchange and Matching Mechanisms in CPS}
Matching theory provides foundational tools for stable and fair resource assignment in multi-agent systems. The Deferred Acceptance algorithm~\cite{gale1962college} and Top-Trading-Cycles (TTC) mechanism~\cite{morrill2024top} ensure stability and strategy-proofness under single ownership. While several TTC extensions expand their theoretical scope, none address post-deviation reassignment in capacity-constrained CPS settings. For instance, \cite{alcalde2011exchange} generalized TTC to account for preference heterogeneity, \cite{morrill2015making} incorporated priority fairness through the concept of justness, and \cite{hakimov2018equitable} introduced an equitable variant that eliminates avoidable envy. More recent efforts adapt TTC to social and networked environments, restricting trades to neighboring agents~\cite{kawasaki2021mechanism,you2022strategy}. However, these frameworks assume static ownership and perfect compliance, limiting adaptability after reassignment. 
The closest work to ours is~\cite{biro2022serial} that studies markets in which each agent may have multiple endowments, whereas we only consider a single one. In contrast to these existing strategies, our capacity-aware TTC framework retains key theoretical guarantees while enabling decentralized, preference-driven reassignment in dynamic, shared-resource CPS environments.

\vspace{-0.1in}
\subsection{Research Gaps}
Existing studies reveal four key gaps: (1) CPS coordination assumes full compliance with no post-deviation reassignment~\cite{kerkman2012car,djavadian2014empirical}; (2) TTC models address only one-to-one ownership, ignoring capacity-constrained settings~\cite{morrill2024top}; (3) one-to-many mechanisms like DAA ensure stability but not Pareto optimality~\cite{gale1962college}; and (4) most CPS schedulers rely on centralized recomputation~\cite{zhao2020allocation,liu2025energy}. We address these through a capacity-aware, TTC framework preserving theoretical guarantees under non-compliance.

\vspace{-0.1in}
\section{System Model and Problem Formulation}\label{sec:model_new}
\subsection{System Model}
We consider a CPS platform coordinating a set of agents $\mathcal{A} = \{a_1, a_2, \dots, a_{|\mathcal{A}|}\}$ denoting the set of autonomous or human-in-the-loop \emph{agents} (e.g., vehicles, UAVs, responders, or service units) that access a set of shared \emph{resources} 
$\mathcal{R} = \{r_1, r_2, \dots, r_{|\mathcal{R}|}\}$ 
(e.g., charging ports or service nodes). Each resource $r_j \in \mathcal{R}$ possesses a finite capacity $q_{j} \in \mathbb{Z}^+$, indicating the number of agents that can be simultaneously served or accommodated. Collectively, these form the quota vector $
\mathbf{q} = (q_1, q_2, \dots, q_{|\mathcal{R}|})$. 
During the assignment process, the quota vector $\mathbf{q}$ is updated
by decrementing $q_j$ whenever an agent is allocated to $r_j$. 

At any decision epoch, an assignment procedure is assumed to be in place to produce an initial one-to-many allocation
$
\omega: \mathcal{A} \rightarrow \mathcal{R} \cup \{\emptyset\},
$
where $\omega(a_i) = r_j$ denotes that agent $a_i$ is assigned to resource $r_j$, and $\omega(a_i)=\emptyset$ represents unassigned agents or idle capacity. Agents may either \emph{comply} with their assigned resource or \emph{deviate} based on their real-time preferences, local context, or behavioral state. We capture such non-compliant agents in $
\mathcal{A}^{\mathrm{nc}} = 
\{\, a_i \in \mathcal{A} : \text{$a_i$ rejects $\omega(a_i)$} \,\},$
and subsequently, the compliant agents are captured in
$\mathcal{A}^{\mathrm{c}} = 
\mathcal{A} \setminus \mathcal{A}^{\mathrm{nc}}$.
Thus, the reassignment mechanism operates exclusively over non-compliant agents $\mathcal{A}^{\mathrm{nc}}$ and the residual capacity vector $\mathbf{q}$, while ensuring that allocations of compliant agents remain untouched.

After identifying the non-compliant agents $\mathcal{A}^{\mathrm{nc}}$, each
agent $a_i \in \mathcal{A}^{\mathrm{nc}}$ is allowed to seek an alternative allocation based on its current operational context. However, not all
resources may be admissible for every agent. We therefore define, for
each such agent, a \emph{feasible resource set}
$\mathcal{R}_i^{\mathrm{feas}} \subseteq \mathcal{R}$,
which captures the resources that $a_i$ can legitimately access under the
system's physical, temporal, or capacity constraints. A resource $r_j$ is included in the feasible set $\mathcal{R}_i^{\mathrm{feas}}$ for an
agent $a_i$ only if:  
(i) $q_j > 0$, ensuring at least one available unit of service;  
(ii) $a_i$ can physically or temporally reach $r_j$ within its operational limits
(e.g., mobility, energy, latency); and  
(iii) $r_j$ satisfies any agent-specific requirements such as charger type,
task compatibility, or priority class. Note that these conditions naturally vary across CPS domains (e.g., an EV requiring a fast charger, a UAV constrained by remaining battery, or a responder restricted by
response-time bounds). Our goal is not to enumerate domain-specific rules, but to ensure that only admissible options after excluding capacity consumed by
compliant agents enter $\mathcal{R}_i^{\mathrm{feas}}$.

Each non-compliant agent $a_i \in \mathcal{A}^{\mathrm{nc}}$ reports a strict
preference ordering over its feasible set $\mathcal{R}_i^{\mathrm{feas}} \subseteq \mathcal{R}$.
This ordering is captured as
$ r_{1} \succ_i r_{2} \succ_i \cdots \succ_i r_{|\mathcal{R}_i^{\mathrm{feas}}|}$,
where $r_{1}$ and $r_{|\mathcal{R}_i^{\mathrm{feas}}|}$ is the most and least preferred respectively.
Together we use $
P(a_i)$ to capture this ordinal relationship of $a_i$. Preferences can be generated using domain-specific factors (e.g., detour, latency, energy, task fit), but the mechanism uses only the ordering over $\mathcal{R}_i^{\mathrm{feas}}$.

\vspace{-0.1in}
\subsection{Mapping Ranks to Satisfaction}
Given a preference ordering $P(a_i)$ over the feasible set $\mathcal{R}_i^{\mathrm{feas}}$, each resource in this ordering is mapped to a numerical rank through a function $
\mathrm{rank}_i : \mathcal{R}_i^{\mathrm{feas}} \rightarrow 
\{1,2,\dots,|\mathcal{R}_i^{\mathrm{feas}}|\}$,
where $\mathrm{rank}_i(r_j)=1$ denotes the most preferred resource according to 
$P(a_i)$. This ranking actually reflects the level of satisfaction of an agent with the final assignment and can be computed in different ways. Some works ~\cite{black1958theory,swain2024m,swain2023dafto} compute the satisfaction score is 
the normalized linear form as shown in Eq. \eqref{eq:linear_score}, where $s_i(r_j)=1$ is for the top-ranked choice and $s_i(r_j)=0$ for the least 
preferred.
\vspace{-0.05in}
\begin{equation}\label{eq:linear_score}
s_i(r_j) 
= 
\frac{|\mathcal{R}_i^{\mathrm{feas}}| - \mathrm{rank}_i(r_j)}{|\mathcal{R}_i^{\mathrm{feas}}|-1}
\end{equation}

\vspace{-0.05in}
Let us consider an agent $a_1$ with preference ordering $P(a_1): r_1 \succ r_2 \succ r_3 \succ r_4$. 
Adopting Eq.~\eqref{eq:linear_score}, the satisfaction values become for different assignments are $s_1(r_1)=1.00$, $s_1(r_2)=0.67$, $s_1(r_3)=0.33$, and $s_1(r_4)=0.00$. 
Thus, assigning $a_1$ to its first, second, third, or fourth choice yields 
100\%, 67\%, 33\%, and 0\% satisfaction respectively. Although the linear score provides a simple normalization of ranks, it fails to capture key aspects of human decision behavior.  
\textit{First}, it assumes uniform sensitivity across ranks, whereas empirical studies~\cite{levy1992introduction, li2025deconstruction, zhang2025self} show that diminishing sensitivity exists, with improvements near the top matter disproportionately more than improvements near the bottom. \textit{Second}, it is reference-blind: an agent’s satisfaction depends only on the assigned rank, not on how that outcome compares to its initial allocation or expectation.  
\textit{Third}, it cannot model threshold effects (e.g., sharp drops beyond a top-$k$ acceptable set), nor can it ensure comparability across agents with differently sized feasible sets.  
These limitations motivate a richer, reference-dependent formulation based on Prospect Theory (PT)~\cite{levy1992introduction} that incorporates diminishing sensitivity and gain-loss asymmetry while preserving ordinal preference structure.

\vspace{-0.1in}
\subsection{Prospect-Theoretic Satisfaction Model}
PT models human evaluation of outcomes relative to a reference point rather than in absolute terms. The generic value function~\cite{kahneman2013prospect}
is defined as Eq.~\eqref{eqn:PT}, where $z$ denotes the gain or loss relative to a reference point, $\alpha$ and $\beta$
capture diminishing sensitivity for gains and losses, and $\lambda$ models loss
aversion.
\vspace{-0.05in}
\begin{equation}\label{eqn:PT}
v(z) \;=\;
\begin{cases}
z^{\alpha}, & z \geq 0, \\[6pt]
-\lambda\,(-z)^{\beta}, & z < 0,
\end{cases}
\qquad 
0<\alpha,\beta \le 1,\;\; \lambda > 1,
\end{equation}

\vspace{-0.05in}
In our CPS setting, each agent compares any reassigned resource to its initial
endowment, which is considered the reference point. 
Let the linear satisfaction from assigning resource $r_j$ to agent $a_i$ be
$s_i(r_j)$ (obtained from Eq.~\eqref{eq:linear_score}), and define the reference satisfaction as $s_i^{\mathrm{ref}} \;=\; s_i\!\big(\omega(a_i)\big)$.
The raw gain from receiving $r_j$ is therefore
$z_i(r_j) \;=\; s_i(r_j) - s_i^{\mathrm{ref}}$.
The overall goal of the re-assignment procedure is to never receive an assignment worse than their endowment; hence $z_i(r_j)\ge 0$ and the
loss region of~\eqref{eqn:PT} is never invoked. 
Thus, only the gain-side curvature
is relevant, yielding the trimmed form $v\big(z_i(r_j)\big) \;=\; \big(z_i(r_j)\big)^{\alpha}$.

Additionally, it is worth noting that the agents may have different reference points $s_i^{\mathrm{ref}}$, implying that the maximum achievable improvement also varies. To ensure comparability across
heterogeneous agents, we normalize by the maximum possible gain: $z_i^{\max} \;=\; 1 - s_i^{\mathrm{ref}}, 
\,\,
\hat{z}_i(r_j) \;=\; \frac{z_i(r_j)}{z_i^{\max}} \in [0,1]$. Thus, the final PT satisfaction obtained by $a_i$ is derived as:
\begin{equation}\label{eqn:PT_satisfaction}
\mathrm{Sat}_i(r_j)
\;=\;
\left(\hat{z}_i(r_j)\right)^{\alpha},
\qquad 0 < \alpha \le 1.
\end{equation}

This formulation preserves ordinal preferences while incorporating behavioral realism through (i) reference dependence, (ii) diminishing sensitivity, and (iii) fair normalization across agents with different feasible sets and
endowment quality.
\vspace{-0.1in}
\subsection{Problem Formulation}
The reassignment procedure aims to boost/maximize the aggregate satisfaction of the non-compliant agents while ensuring feasibility and respecting all capacity constraints. 
For each non-compliant agent $a_i \in \mathcal{A}^{\mathrm{nc}}$ and resource $r_j \in \mathcal{R}$, let $y_{i,j} \in \{0,1\}$ be a binary indicator variable capturing the reassignment of  $a_i$ to $r_j$.
The objective as expressed in Eq. \eqref{obj:max} is to maximize the aggregate PT satisfaction. Constraint~\eqref{cons:agent} ensures that each non-compliant agent is assigned to at most one resource. On the other hand, Constraint~\eqref{cons:capacity} enforces that no resource exceeds its capacity. Constraint~\eqref{cons:feasible} restricts assignments to only those resources that are feasible for each agent. Constraint~\eqref{cons:IR} guarantees that the non-complaint agents are never worse-off by preventing assignments that yield lower satisfaction than the agent’s initial endowment. Constraint~\eqref{cons:binary} specifies that all assignment decisions are binary.

\vspace{-0.05in}
\begin{align}
maximize
&\quad
\sum_{a_i \in \mathcal{A}^{\mathrm{nc}}}
\;\sum_{r_j \in \mathcal{R}}
\mathrm{Sat}_i(r_j) \times y_{i,j}
\label{obj:max}
\\[4pt]
\text{s.t.} \quad
&\sum_{r_j \in \mathcal{R}} y_{i,j}
\;\le\; 1,
\qquad
\forall\, a_i \in \mathcal{A}^{\mathrm{nc}}
\label{cons:agent}
\\[4pt]
&\sum_{a_i \in \mathcal{A}^{\mathrm{nc}}} y_{i,j}
\;\le\; q_j,
\qquad
\forall\, r_j \in \mathcal{R}
\label{cons:capacity}
\\[4pt]
&y_{i,j} = 0,
\qquad
\forall\, a_i \in \mathcal{A}^{\mathrm{nc}},~
\forall\, r_j \notin \mathcal{R}_i^{\mathrm{feas}} \subseteq \mathcal{R}
\label{cons:feasible}
\\[4pt]
&\mathrm{Sat}_i(r_j)\times y_{i,j}
\;\ge\;
\mathrm{Sat}_i(\omega(a_i)),
\qquad
\forall\, a_i \in \mathcal{A}^{\mathrm{nc}}
\label{cons:IR}
\\[4pt]
&y_{i,j} \in \{0,1\},
\qquad
\forall\, a_i \in \mathcal{A}^{\mathrm{nc}},~
\forall\, r_j \in \mathcal{R}
\label{cons:binary}
\end{align}

\vspace{-0.05in}
Instead of relying on centralized optimization, the next section demonstrates how a capacity-aware extension of the TTC mechanism produces these assignments in a preference-consistent manner while respecting the constraints.

\begin{figure*}[!htbp]
\centering
\captionsetup[sub]{font=small, skip=2pt}
\setlength{\tabcolsep}{2pt}
\renewcommand{\arraystretch}{1.05}
\subcaptionbox{Endowments and preferences. Quotas $(q_{1}=1,q_{2}=1,q_{3}=1)$. \label{tab:ttc-example}}[0.24\textwidth]{%
\centering\small
\begin{tabular}{|c|c|c|}
\hline
\bf Agent & \bf $\omega(a_i)$ & \bf $P(a_i)$ \\ \hline
$a_1$ & $r_1$ & $r_2 \succ r_1 \succ r_3$ \\
$a_2$ & $r_2$ & $r_1 \succ r_2 \succ r_3$ \\
$a_3$ & $r_3$ & $r_3 \succ r_1 \succ r_2$ \\ \hline
\end{tabular}
}
\hfill
\subcaptionbox{TTC directed graph. \label{fig:TTC_classical_graph}}[0.24\textwidth]{%
\centering
\includegraphics[width=\linewidth]{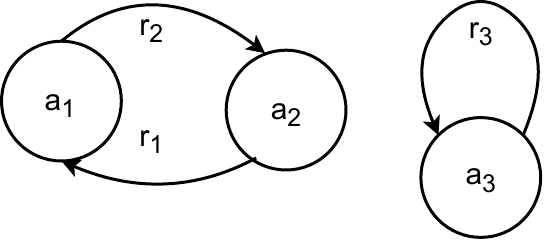}
}
\hfill
\subcaptionbox{Partial endowments; residual at $r_1$. Quotas $(q_{1}=2,q_{2}=1,q_{3}=1)$. \label{tab:caseB}}[0.24\textwidth]{%
\centering\small
\begin{tabular}{|c|c|c|}
\hline
\bf Agent & \bf $\omega(a_i)$ & \bf $P(a_i)$ \\ \hline
$a_1$ & $r_1$ & $r_2 \succ r_1 \succ r_3$ \\
$a_2$ & $r_2$ & $r_1 \succ r_2 \succ r_3$ \\
$a_3$ & $r_3$ & $r_2 \succ r_1 \succ r_3$ \\ \hline
\end{tabular}
}
\hfill
\subcaptionbox{TTC on setup in Table~\ref{tab:caseB}. \label{fig:TTC_classical_graph_case1_new}}[0.24\textwidth]{%
\centering
\includegraphics[width=\linewidth]{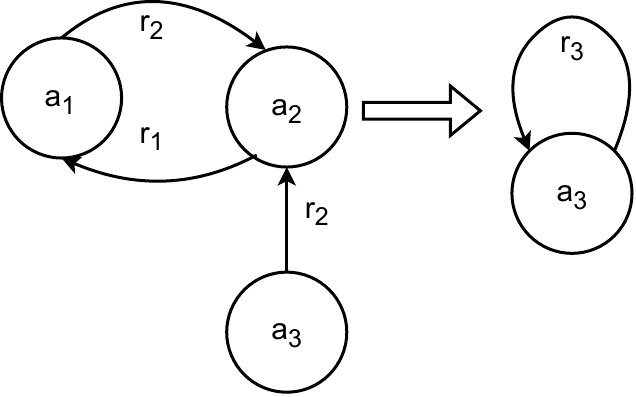}
}
\vspace{-0.15in}
\caption{TTC example instances with full and partial endowments.}
\vspace{-0.1in}
\label{fig:ttc-row}
\end{figure*}

\begin{figure*}[!htbp]
\centering
\captionsetup[sub]{font=small, skip=2pt}
\setlength{\tabcolsep}{2pt}
\renewcommand{\arraystretch}{1.05}
\subcaptionbox{Preference list and initial endowments. $(q_{1}=1,q_{2}=2,q_{3}=1)$. \label{tab:Example_table_caseA1}}[0.24\textwidth]{%
\centering\small
\begin{tabular}{|c|c|c|}
\hline
\textbf{Agent} & \textbf{$\omega(a_i)$} & \textbf{$P(a_i)$} \\ \hline
$a_1$ & $r_1$ & $r_3 \succ r_1 \succ r_2$ \\
$a_2$ & $r_2$ & $r_1 \succ r_2 \succ r_3$ \\
$a_3$ & $r_3$ & $r_2 \succ r_1 \succ r_3$ \\
$a_4$ & $r_2$ & $r_3 \succ r_1 \succ r_2$ \\ \hline
\end{tabular}
}
\subcaptionbox{Augmented directed graph. \label{fig:TTC_example}}[0.24\textwidth]{%
\centering
\includegraphics[width=\linewidth]{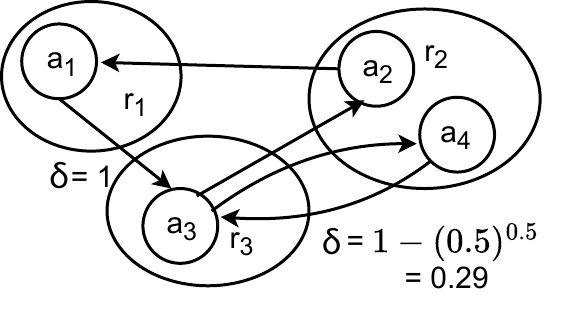}
}
\subcaptionbox{Shorter cycle resolution. ($\delta = 0.29$) \label{tab:Example_1_satisfaction_short}}[0.26\textwidth]{%
\centering\small
\begin{tabular}{|c|c|c|}
\hline
\textbf{Agent} & \textbf{Assignment} & $\hat z_i^{\alpha}$ \\ \hline
$a_1$ & $r_1$ & $0$ \\
$a_2$ & $r_2$ & $0$ \\
$a_3$ & $r_2$ & $1$ \\
$a_4$ & $r_3$ & $1$ \\ \hline
\end{tabular}
}
\subcaptionbox{Longer cycle resolution. ($\delta = 1$)  \label{tab:Example_1_satisfaction_long}}[0.24\textwidth]{%
\centering\small
\begin{tabular}{|c|c|c|}
\hline
\textbf{Agent} & \textbf{Assignment} & $\hat z_i^{\alpha}$ \\ \hline
$a_1$ & $r_3$ & $1$ \\
$a_2$ & $r_1$ & $1$ \\
$a_3$ & $r_2$ & $1$ \\
$a_4$ & $r_2$ & $0$ \\ \hline
\end{tabular}
}
\vspace{-0.15in}
\caption{Case A, example 1: Preferences, augmented graph, and the final satisfactions ($\alpha = 0.5$).}
\vspace{-0.1in}
\label{fig:caseA1}
\end{figure*}

\begin{figure*}[!htbp]
\centering
\captionsetup[sub]{font=small, skip=2pt}
\setlength{\tabcolsep}{2pt}
\renewcommand{\arraystretch}{1.05}
\subcaptionbox{Preference list and initial endowments. $(q_{1}=1,q_{2}=2,q_{3}=1)$. \label{tab:Example_table_caseA2}}[0.24\textwidth]{%
\centering\small
\begin{tabular}{|c|c|c|}
\hline
\textbf{Agent} & \textbf{$\omega(a_i)$} & \textbf{$P(a_i)$} \\ \hline
$a_1$ & $r_1$ & $r_3 \succ r_2 \succ r_1$ \\
$a_2$ & $r_2$ & $r_1 \succ r_2 \succ r_3$ \\
$a_3$ & $r_3$ & $r_2 \succ r_1 \succ r_3$ \\
$a_4$ & $r_2$ & $r_3 \succ r_2 \succ r_1$ \\ \hline
\end{tabular}
}
\subcaptionbox{Augmented directed graph. \label{fig:TTC_example2}}[0.24\textwidth]{%
\centering
\includegraphics[width=\linewidth]{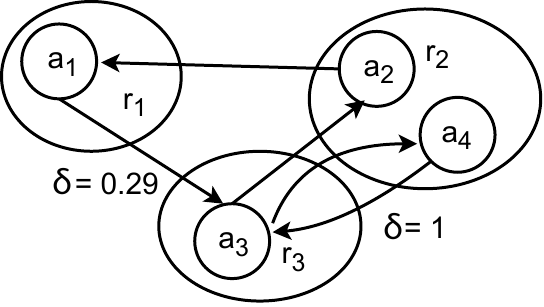}
}
\subcaptionbox{Shorter cycle resolution. ($\delta = 1$) \label{tab:Example_2_satisfaction_short_new}}[0.24\textwidth]{%
\centering\small
\begin{tabular}{|c|c|c|}
\hline
\textbf{Agent} & \textbf{Assignment} & $\hat z_i^{\alpha}$ \\ \hline
$a_1$ & $r_2$ & $0.71$ \\
$a_2$ & $r_1$ & $1$ \\
$a_3$ & $r_2$ & $1$ \\
$a_4$ & $r_3$ & $1$ \\ \hline
\end{tabular}
}
\subcaptionbox{Longer cycle resolution. ($\delta = 0.29$) \label{tab:Example_2_satisfaction_long_new}}[0.26\textwidth]{%
\centering\small
\begin{tabular}{|c|c|c|}
\hline
\textbf{Agent} & \textbf{Assignment} & $\hat z_i^{\alpha}$ \\ \hline
$a_1$ & $r_3$ & $1$ \\
$a_2$ & $r_1$ & $1$ \\
$a_3$ & $r_2$ & $1$ \\
$a_4$ & $r_2$ & $0$ \\ \hline
\end{tabular}
}
\vspace{-0.1in}
\caption{Case A, example 2: Preferences, augmented graph, and the final satisfactions ($\alpha = 0.5$).}
\vspace{-0.2in}
\label{fig:caseA2}
\end{figure*}

\section{Proposed Solution}\label{sec:solution_Approach_new}
\vspace{-0.05in}
\subsection{Resource Allocation as Assignment Problem}
\begin{definition}[\textbf{Preferences}]
Each non-compliant agent $a_i \in \mathcal{A}^{\mathrm{nc}}$ reports a strict preference ordering $\succ_{a_i}$ over its feasible option set denoted by $\mathcal{R}_i^{\mathrm{feas}} \subseteq \mathcal{R}$.
\end{definition}

\vspace{-0.05in}
\begin{definition}[\textbf{Assignment}]
An assignment is a one-to-many mapping 
$\mu:\mathcal{A}^{\mathrm{nc}} \cup \mathcal{R} \rightarrow 
2^{\mathcal{A}^{\mathrm{nc}} } \, \cup \, 2^\mathcal{R}$
satisfying:
    (i) $\forall\, a_i\in\mathcal{A}^{\mathrm{nc}},\ |\mu(a_i)|\le 1$ and $\mu(a_i)\subseteq \mathcal{R}$,
    (ii) $\forall\, r_j\in\mathcal{R},\ |\mu(r_j)| \le q_j$ and $\mu(r_j)\subseteq \mathcal{A}^{\mathrm{nc}}$,
    (iii) $a_i \in \mu(r_j)$ \ $\Longleftrightarrow$ \ $r_j \in \mu(a_i)$.
\end{definition}

\vspace{-0.05in}
\begin{definition}[\textbf{Pareto Optimality}]
An assignment $\mu$ is \emph{Pareto optimal} if there exists no feasible assignment $\mu'$ such that:
 $\exists\, a_i$ with $\mu'(a_i) \succ_{a_i} \mu(a_i)$; and
 $\forall\, a_{i'} \ne a_i$, $\mu'(a_{i'}) \succeq_{a_{i'}} \mu(a_{i'})$.
Thus, $\mu'$ makes at least one agent strictly better off without making others worse off.
\end{definition}

\vspace{-0.05in}
\begin{definition}[\textbf{Individual Rationality}]
With $\omega$ capturing the initial endowment mapping 
$\omega:\mathcal{A}^{\mathrm{nc}} \to \mathcal{R}\cup\{\emptyset\}$.
An assignment $\mu$ is \emph{individually rational} if
$\mu(a_i) \succeq_{a_i} \omega(a_i), \forall\, a_i\in\mathcal{A}^{\mathrm{nc}}.$
%
Thus, no agent is worse off than its initial endowment.
\end{definition}

\vspace{-0.05in}
\begin{definition}[\textbf{Core Stability}]\label{def:core_stability}
Let $\mathcal{A}'\subseteq\mathcal{A}^{\mathrm{nc}}$ and 
$\mathcal{R}' = \{\omega(a_i)\mid a_i\in\mathcal{A}'\}$ be their endowed resources.  
An assignment $\mu$ is \emph{core stable} if there is no coalition $(\mathcal{A}',\mathcal{R}')$ and feasible reassignment $\mu'$ such that
$\mu'(a_i) \succ_{a_i} \mu(a_i), \forall\, a_i\in\mathcal{A}'.$
%
That is, no subset of agents can mutually reassign only their endowed resources and all become strictly better off.
\end{definition}

\begin{definition}[\textbf{Strategy-proofness}]
A mechanism is \emph{strategy-proof} if no agent can obtain a strictly better
assignment by misreporting. Let $\mu$ be the allocation under the true
preference profile $(\succ_{a_i})_{a_i\in\mathcal{A}^{\mathrm{nc}}}$, and $\mu'$ be the allocation when some agent $a_i$ misreports $\succ_{a_i}'$, while others
report truthfully. The mechanism is strategy-proof if, for every agent $a_i$, $\mu(a_i) \succeq_{a_i} \mu'(a_i)$, where $\succeq_{a_i}$ is $a_i$'s true preference. 
\end{definition}
\vspace{-0.1in}
\subsection{TTC: Limitations in Constrained CPS}
To illustrate the working of TTC, consider a simple one-to-one exchange market with three agents $\mathcal{A}=\{a_1,a_2,a_3\}$ and three resources $\mathcal{R}=\{r_1,r_2,r_3\}$, where each has unit capacity ($q_j=1$). Each agent's endowments and preferences is shown in Table \ref{tab:ttc-example}.
Using the instance in Table~\ref{tab:ttc-example}, the TTC procedure begins constructing a graph, where the agents ($a_i$) are represented by vertices in a graph. In the first step (graph construction), each agent points to the owner of its most preferred resource via a directed edge. For example, $a_1$ points to $a_2$, the owner of its most preferred resource $r_2$. This step leads to addition of the directed edges $(a_1 \to a_2)$, $(a_2 \to a_1)$ , and $(a_3 \to a_3)$ (self-loop) as shown in Fig. \ref{fig:TTC_classical_graph}. 
At the end of graph construction, the graph must contain at least one cycle~\cite{morrill2024top}.
For instance, Fig. \ref{fig:TTC_classical_graph} has two cycles $(a_1 \to a_2 \to a_1)$ and $(a_3 \to a_3)$. In the next step (cycle resolution), the trade is conducted among the agents in a cycle by reassigning resources to the individuals who pointed to the current owner. For instance, by resolving the first cycle, $a_1$ receives $r_2$ and $a_2$ receives $r_1$; both agents and resources are then removed from further consideration. Similarly, the other cycle is also resolved. Therefore the final allocation is $\mu(a_1)=r_2$, $\mu(a_2)=r_1$, and $\mu(a_3)=r_3$.
The outcome of TTC in the above example satisfies all the properties established previously~\cite{khoshdel2025cg}. 

On the other hand, let Table~\ref{tab:caseB} represent another setup, where the capacity of resource $r_1$ is $2$, and the preference of $a_3$ is $r_2 \succ r_1 \succ r_3 $. In this changed setup, the classical TTC leads to the final assignment $\mu(a_1)=r_2$, $\mu(a_2)=r_1$, and $\mu(a_3)=r_3$ as depicted in Fig. \ref{fig:TTC_classical_graph_case1_new}.
Since TTC operates only on endowed units, the spare slot at $r_1$ remains ``invisible'', causing $a_3$ to be assigned to $r_3$ even though it strictly prefers $r_1$, where a slot is still available.
This assignment is \textit{core-stable} in the endowment sense, since the spare capacity at $r_1$ has no owner and thus cannot participate in any blocking coalition. However, in capacity-constrained CPS, this notion of \textit{core-stability} is inefficient and misaligned, as the unowned capacity represents real, usable supply.
Such an assignment can be improved by a direct assignment of $a_3$ to $r_1$, without harming any other agent, and thus the current assignment violates the \textit{Pareto optimality}.
Although a direct reassignment resolves the issue in this simple example, the situation becomes more complex when multiple agents compete for the same available resource, necessitating novel approaches to assignment.
\begin{figure*}[h]
\centering
\captionsetup[sub]{font=small, skip=2pt}
\setlength{\tabcolsep}{2pt}
\renewcommand{\arraystretch}{1.05}
\subcaptionbox{Case B: Endowments and preferences. Quotas $(q_{1}=2,q_{2}=2,q_{3}=1,q_{4}=1)$. \label{tab:caseB_table}}[0.24\textwidth]{%
\centering\small
\begin{tabular}{|c|c|c|}
\hline
\textbf{Agent} & \textbf{$\omega(a_i)$} & \textbf{$P(a_i)$} \\ \hline
$a_1$ & $r_1$ & $r_2 \succ r_1 \succ r_3 \succ r_4$ \\
$a_2$ & $r_2$ & $r_1 \succ r_2 \succ r_3 \succ r_4$ \\
$a_3$ & $r_3$ & $r_1 \succ r_4 \succ r_2 \succ r_3$ \\
$a_4$ & $r_2$ & $r_4 \succ r_3 \succ r_1 \succ r_2$ \\
$a_5$ & $r_4$ & $r_3 \succ r_2 \succ r_4 \succ r_1$ \\ \hline
\end{tabular}
}
\subcaptionbox{Case B: Augmented directed graph. \label{fig:TTC_classical_graph_case2}}[0.26\textwidth]{%
\centering
\includegraphics[width=\linewidth]{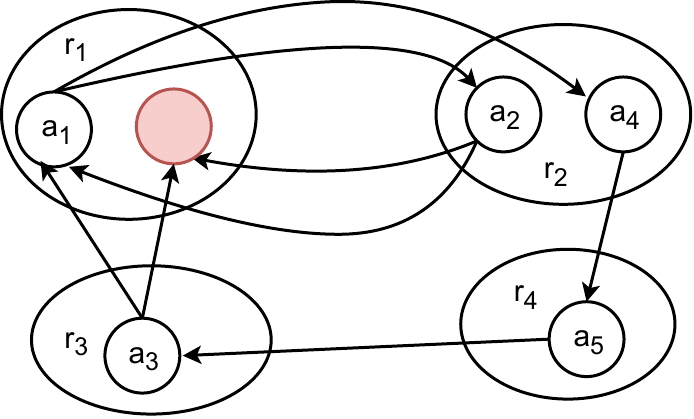}
}
\subcaptionbox{Path to a virtual vertex. \label{fig:caseB_chain}}[0.21\textwidth]{%
\centering
\includegraphics[width=\linewidth]{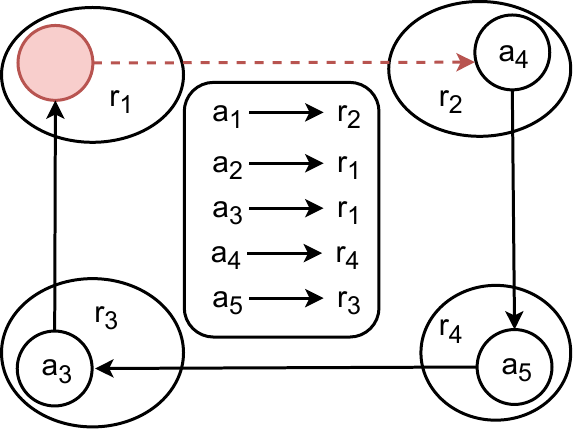}
}
\subcaptionbox{Assignments and satisfactions. \label{tab:Example_2_satisfaction_short}}[0.24\textwidth]{%
\centering\small
\begin{tabular}{|c|c|c|}
\hline
\textbf{Agent} & \textbf{Assignment} & $\hat z_i^{\alpha}$ \\ \hline
$a_1$ & $r_2$ & $1$ \\
$a_2$ & $r_1$ & $1$ \\
$a_3$ & $r_1$ & $1$ \\
$a_4$ & $r_1$ & $1$ \\ 
$a_5$ & $r_3$ & $1$ \\\hline
\end{tabular}
}
\vspace{-0.2in}
\caption{Case B: Endowments, augmented graph, and two cycle-resolution sequences.}
\vspace{-0.2in}
\label{fig:caseB_row}
\end{figure*}
\vspace{-0.1in}
\subsection{Solution Approach}\label{sec:solution_app}
\noindent \textbf{Case A (Co-ownership):} We begin by modifying the standard TTC graph construction~\cite{morrill2024top}. Each resource is treated as a meta-node containing its initially endowed agents as interior nodes. For example, in Fig.~\ref{fig:caseA1}, resource $r_2$ becomes a meta-node whose components are agents $a_2$ and $a_4$, as listed in Table~\ref{tab:Example_table_caseA1}. The same procedure applies to all resources. Next, each agent draws directed edges toward the owners of its top preference. This is the key departure from classical TTC, which restricts agents to a single outgoing edge under atomic (one-to-one) ownership. In our setting, resources may be co-owned. For instance, $r_2$ is jointly owned by $a_2$ and $a_4$ (Fig.~\ref{fig:caseA1}). Therefore, if an agent’s preferred resource is co-owned, as in the case of $a_3$, we draw one edge to each owner ($a_2$ and $a_4$). For single-owner resources, we draw one edge following the classical TTC rule. The resulting directed graph is shown in Fig.~\ref{fig:TTC_example}.

We now turn to cycle resolution, which becomes substantially more intricate in our setting. In contrast to classical TTC, where each agent has exactly one outgoing edge and all cycles are therefore vertex-disjoint, we allow multiple outgoing edges, which enables the formation of overlapping cycles. As illustrated in 
Fig.~\ref{fig:TTC_example}, two cycles coexist: a shorter cycle $(a_3 \to a_4 \to a_3)$ and a longer cycle $(a_1 \to a_3 \to a_2 \to a_1)$, with agent $a_3$ appearing in both. This overlap makes the resolution order decisive. Resolving the shorter cycle first leads to an assignment (Table~\ref{tab:Example_1_satisfaction_short}) with a cumulative satisfaction score of $2$.
Whereas prioritizing the longer cycle produces the assignment shown in Table~\ref{tab:Example_1_satisfaction_long}, and it leads to a higher cumulative satisfaction than the previous assignment.
However, if the preference lists are changed slightly as Table~\ref{tab:Example_table_caseA2}, the same cycle resolution for the augmented graph (Fig.~\ref{fig:TTC_example2}) provides an opposite result. From the assignment Tables~\ref{tab:Example_2_satisfaction_short_new} and ~\ref{tab:Example_2_satisfaction_long_new}, it is evident that in this example, resolving the shorter cycle first leads to a better total satisfaction.
These observations underscore that, unlike classical TTC, the cycle-resolution order has a critical impact on overall satisfaction, and finding the optimal ordering is non-trivial. Overlapping cycles, competing improvements, and interdependencies across edges create a combinatorial search space, which is challenging to solve. 

Instead of finding a global cycle resolution order through exhaustive search, we develop a heuristic that optimizes satisfaction locally in each round of TTC. 
We notice that resolving a cycle $c_1$ among multiple cycles $\{c_1, \dots, c_n\}$ overlapping at vertex $a_i$ impacts the satisfaction of its immediate predecessor vertices $a_k \in \{c_2, \dots, c_n\}$ as resolving $c_1$ removes the edge $(a_k \to a_i)$ and forces $a_k$ to move to its next choice in the preference list. 
We capture this impact on the satisfaction of immediate predecessors through the \textit{minimum satisfaction loss} $\delta$, defined as the difference between the satisfaction of obtaining the currently chosen resource and that of obtaining the next preferred resource. Therefore, prioritizing the resolution of cycles with larger $\delta$ minimizes the satisfaction loss of the predecessors and locally maximizes the satisfaction. In Fig.~\ref{fig:caseA1} and \ref{fig:caseA2}, resolving the cycles according to $\delta$ yields higher cumulative satisfaction.

\vspace{0.05in}

\noindent \textbf{Case B (No ownership/ Partial ownership):} This case considers a scenario where the resources are not at capacity, and there is at least a vacant slot which is preferred by some agent. As the classical TTC cannot handle empty slots, we consider a virtual agent owner per empty slot (indicated by a red circle in Fig. \ref{fig:TTC_classical_graph_case2}) to force the involvement of empty slots in the trading. 
Consequently, by following the graph construction step of TTC from Table~\ref{tab:caseB_table}, we get a directed graph (Fig. \ref{fig:TTC_classical_graph_case2}) with two cycles, $(a_1 \to a_2 \to a_1)$ and $(a_1 \to a_4 \to a_5 \to a_3 \to a_1)$.

However, since virtual owners have no preferences, they have no outgoing edges and therefore cannot be part of any cycle. Hence, regardless of the order of resolution, resolving all cycles results in a directed acyclic graph (DAG) whose paths terminate at virtual agent vertices. In our example, resolving the cycle $(a_1 \to a_2 \to a_1)$ yields the path $a_4 \to a_5 \to a_3 \to a^v$ (Fig.~\ref{fig:caseB_chain}). The assignment trades along any path ending at a virtual vertex can then be implemented in one step by creating a cycle, that is, by adding an edge from the virtual vertex $a^v$ to the starting vertex of the path $a_4$ in the example.

If there exist multiple paths to a virtual vertex, they overlap at least at the virtual vertex, and their resolution order can be determined by computing the \textit{minimum satisfaction loss} $\delta$ of the immediate predecessors of the overlapping vertex, similar to case A. For both Cases A and B, if there exist multiple overlapping vertices (Fig.~\ref{fig:multiple_overlap}), the ordering can be determined by taking the sum of the minimum satisfaction losses.

\begin{figure}[!htbp]
\centering
\captionsetup[sub]{font=small, skip=2pt}
\setlength{\tabcolsep}{2pt}
\renewcommand{\arraystretch}{1.05}
\subcaptionbox{Example 3. Quotas $(q_{1}=1,q_{2}=2,q_{3}=1)$. \label{tab:Example3_table}}[0.24\textwidth]{%
\centering\small
\begin{tabular}{|c|c|c|}
\hline
\textbf{Agent} & \textbf{$\omega(a_i)$} & \textbf{$P(a_i)$} \\ \hline
$a_1$ & $r_1$ & $r_2 \succ r_3 \succ r_1$ \\
$a_2$ & $r_2$ & $r_1 \succ r_3 \succ r_2$ \\
$a_3$ & $r_3$ & $r_2 \succ r_1 \succ r_3$ \\
$a_4$ & $r_2$ & $r_3 \succ r_1 \succ r_2$ \\ \hline
\end{tabular}
}
\subcaptionbox{TTC-graph for Table~\ref{tab:Example3_table}. \label{fig:TTC_example3}}[0.23\textwidth]{%
\centering
\includegraphics[width=\linewidth]{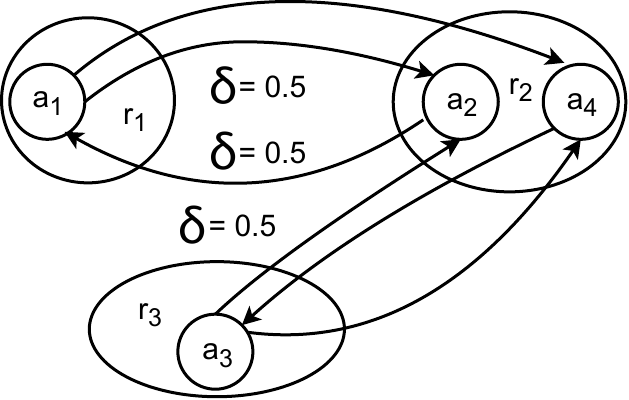}
}\\
\subcaptionbox{Example 4. Quotas $(q_{1}=1,q_{2}=2,q_{3}=1,q_{4}=1)$. \label{tab:ExampleB1_table}}[0.24\textwidth]{%
\centering\small
\begin{tabular}{|c|c|c|}
\hline
\textbf{Agent} & \textbf{$\omega(a_i)$} & \textbf{$P(a_i)$} \\ \hline
$a_1$ & $r_1$ & $r_4 \succ r_3 \succ r_1 \succ r_2$ \\
$a_2$ & $r_2$ & $r_1 \succ r_2 \succ r_3 \succ r_4$ \\
$a_3$ & $r_3$ & $r_2 \succ r_1 \succ r_3 \succ r_4$ \\
$a_4$ & $r_2$ & $r_4 \succ r_3 \succ r_1 \succ r_2$ \\ \hline
\end{tabular}
}
\subcaptionbox{TTC-graph for Table~\ref{tab:ExampleB1_table}. \label{fig:TTC_example_case_B}}[0.2\textwidth]{%
\centering
\includegraphics[width=\linewidth]{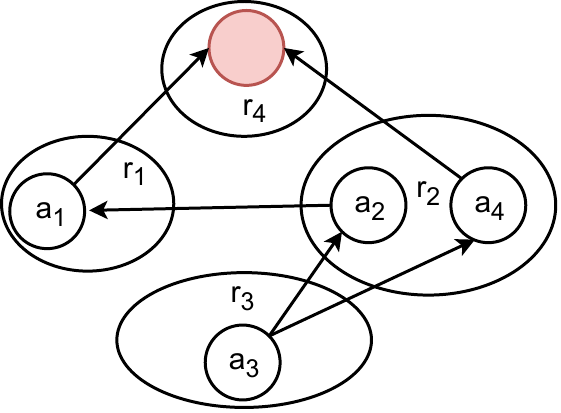}
}
\vspace{-0.15in}
\caption{TTC-graph for multiple overlapping vertices.}
\vspace{-0.15in}
\label{fig:multiple_overlap}
\end{figure}
\vspace{-0.1in}
\noindent \textbf{ReACT-TTC:} To efficiently handle Cases A, B, and their combinations, here we formally design ReACT-TTC (Algorithm~\ref{algo:reactTTC}). It consists of three major steps that repeat until all agents have obtained their final resource assignment.
\vspace{0.05in}

\noindent \textbf{\textit{Step 1 (Graph Construction).}} At each iteration, this step considers all the agents without a final assignment as a vertex set $V$.
Next, it constructs directed edges from each agent $a_i$ to the owner of its most preferred resource if the agent does not have any existing outgoing edge (i.e., the outgoing neighbor set $N^-(a_i) = \emptyset$). If a resource with empty slots is preferred at this iteration, virtual owner vertices are added (Algorithm~\ref{algo:reactTTC} Line 7-11). For each edge $(a_i \to a_k)$, the edge weight is assigned to the minimum satisfaction loss computed as $\delta(a_i, r_j) = \hat z_i^{\alpha}(r_j) - \hat z_i^{\alpha}(r_j^{next})$, where $r_j, r_j^{next}$ are the current and next resource preference of $a_i$,  and $z_i^{\alpha}(r_j), \hat z_i^{\alpha}(r_j^{next})$ are their related satisfactions, respectively. The graph $G(V,E)$ is used in the algorithm, where $E$ denotes the set of edges. 
\vspace{0.05in}

\noindent \textbf{\textit{Step 2 (Complete Cycle Resolution). }} This step finds the simple cycles~\cite{johnson1975finding} in the directed graph $G$ and sorts them in decreasing order of their total satisfaction loss computed as $\sum_{a_i \in \mathcal{O}} \delta(a_i,r_j)$, where $\mathcal{O}$ is the set of overlapping vertices in a cycle $c$. Then the cycles are resolved using Algorithm~\ref{algo:resolve_cycle} by following the sorted order. For each edge $(a_i \to a_{i'})$ in a cycle, the resolution stage assigns the current resource of $a_{i'}$ to $a_i$ as the final assignment $\mu(a_i)$ and updates the current assignment of the resource (Algorithm~\ref{algo:resolve_cycle} Line 4-10). If a resource obtains final assignments for all its slots, it is removed from the preference lists of all agents. Finally, after resolving a cycle $c$, all of its vertices are removed from the graph.

\noindent \textbf{\textit{Step 3 (Incomplete Cycle Resolution). }}
Completing all cycles in Step 2 makes $G$ a DAG. Step 3 finds all the directed paths ending at the virtual vertices and sorts them in decreasing order of their total satisfaction loss. Next, each path from the sorted list of paths is processed, and the end vertex (virtual) is connected with the starting vertex of the path. It leads to a cycle, which is resolved using Algorithm~\ref{algo:resolve_cycle} to obtain the final assignments of the agents along the path. As in Step 2, at the end, all vertices on the path are removed from the graph.

These steps continue until all agents receive their final assignments or their preference lists are exhausted. 

\noindent \textbf{Time Complexity Analysis. } Algorithm~\ref{algo:reactTTC} Step 1 takes $O(q_{max} \cdot |\mathcal{A}^{\mathrm{nc}}|)$ time, where $q_{max}$ is the maximum capacity of any resource and $\mathcal{A}^{\mathrm{nc}}$ is the total agents in the system. Step 1 creates a graph of $O(|\mathcal{A}^{\mathrm{nc}}|)$ vertices and $O(q_{max} \cdot |\mathcal{A}^{\mathrm{nc}}|)$ edges.
Therefore, finding all the cycles in a round of ReACT-TTC can take $O(|V| + |E|)(|\mathcal{C}|) = O(q_{max} \cdot |\mathcal{A}^{\mathrm{nc}}|^2)$ time. Finding the satisfaction loss and sorting the cycles takes $O(|\mathcal{A}^{\mathrm{nc}}| \cdot \log(|\mathcal{A}^{\mathrm{nc}}|))$ time. Resolving a cycle takes $O(|\mathcal{A}^{\mathrm{nc}}|)$ time. Step 3 requires the same time complexity as Step 2. If the termination of ReACT-TTC requires $|\mathcal{A}^{\mathrm{nc}}|$ rounds, the worst case complexity of the algorithm will be $O(|\mathcal{A}^{\mathrm{nc}}| \cdot (q_{max} \cdot |\mathcal{A}^{\mathrm{nc}}|^2 +|\mathcal{A}^{\mathrm{nc}}| \cdot \log(|\mathcal{A}^{\mathrm{nc}}|) = O(q_{max} \cdot |\mathcal{A}^{\mathrm{nc}}|^3)$.

\begin{algorithm}
\caption{ReACT-TTC}
\label{algo:reactTTC}
\DontPrintSemicolon
\small
\KwIn{$\mathcal{A}^{\mathrm{nc}}, \mathcal{R}$}
\KwOut{$\mu:\mathcal{A}^{\mathrm{nc}} \cup \mathcal{R} \rightarrow 
2^{\mathcal{A}^{\mathrm{nc}} \cup \mathcal{R}}$}
Initialize a graph $G(V,E)$ such that $V=\mathcal{A}^{\mathrm{nc}}$ and $E=\emptyset$\;
Initialize a set of virtual vertices $V' = \emptyset$\;

\While{$True$}{
\tcc{Step 1: Graph Construction}
\For{$a_i \in V\setminus V'$}{
        \If{$N^-(a_i) = \emptyset$ and $P(a_i) \neq \emptyset$}{
            Preferred resource $r_j \gets \texttt{Extract}(P(a_i))$\;

            \tcc{Create Virtual Vertices}
            Available quota $\varepsilon \gets q_j - |\omega(r_j)|$\;
            \If{$\varepsilon > 0$}{
                Create virtual vertices $\{a^r_1, \dots, a^r_{\varepsilon}\}$\; 
                Add virtual vertices to $V, V', $ and $\omega(r_j)$\;
                $\omega(a^r_k) \gets r_j \, \, \forall a^r_k \in \{a^r_1, \dots, a^r_{\varepsilon}\}$\;
            }
            \tcc{Add Directed Weighted Edges}
            \For{$a_{k} \in \omega(r_j)$}{
                Add directed edge $(a_i\rightarrow a_k)$ with edge weight $\delta(a_i, r_j)$ to $E$\;
            }
            
        }
    }
\tcc{Step 2: Complete Cycle Resolution}
$\mathcal{C} \gets $ \texttt{Find\_simple\_cycles}$(G(V,E))$\;
Sort the cycles in $\mathcal{C}$ in decreasing order of their total satisfaction loss\; 
\For{cycle $c$ in sorted $ \mathcal{C}$}{
    \If{any vertex in $c$ does not exist in $V$}{
        Skip\;
    }
    \texttt{Resolve\_cycle}$(c, \mu, \omega)$\;
    Remove all vertices in $c$ from $G(V,E)$\;
}
\tcc{Step 3: Incomplete Cycle Resolution}
$\mathcal{L} \gets$ Find directed paths ending at virtual vertices $V'$ in directed acyclic graph $G(V,E)$\;
Sort the paths in $\mathcal{L}$ in decreasing order of their total satisfaction loss\; 
\For{paths $\iota$ in sorted $ \mathcal{L}$}{
    \If{any vertex in $\iota$ does not exist in $V$}{
        Skip\;
    }
    $c \gets$ Create a cycle by adding a directed edge from the last vertex to the first vertex of path $\iota$\;
    \texttt{Resolve\_cycle}$(c, \mu, \omega)$\;
    Remove all vertices in $c$ from $G(V,E)$\;
}
\tcc{Terminate Iterations}
\If{$|V \setminus V'|=0$ or $P(a_i)=\emptyset \quad \forall a_i \in \mathcal{A}^{nc}$}{
    break to outer loop\;
}
}
\end{algorithm}

\begin{algorithm}
\caption{Resolve\_cycle$(c, \mu, \omega)$}
\label{algo:resolve_cycle}
\DontPrintSemicolon

\For{each edge $(a_i\rightarrow a_{i'}) \in c$}{
    $r_j \gets \omega(a_j)$\;
    $r_{j'} \gets \omega(a_{i'})$\;

    \If{$a_i$ is not a virtual vertex}{
        $\mu(a_i) \gets r_{j'}$\;
        $\mu(r_{j'}) \gets \mu(r_{j'}) \cup \{a_i\}$\;
        $\omega(r_{j'}) \gets \omega(r_{j'}) \setminus \{a_{i'}\}$\;
        Update quota $q_{j'} \gets q_{j'} - 1$\;
    }
    \Else{
        $\omega(r_{j'}) \gets \omega(r_{j'}) \setminus \{a_{i'}\} \cup \{a_i\}$\;
    }
     
     \If{$q_{j'} = 0$}{
        Remove $r_{j'}$ from all preference lists $P$\; 
     }
}
\end{algorithm}
\vspace{-0.05in}
\section{Theoretical Analysis}\label{sec:theory}
\begin{theorem}[\textbf{Termination}]\label{thm:termination_case1}
TTC terminates after a finite number of rounds and returns a feasible assignment.
\end{theorem}
\vspace{-0.15in}
\begin{proof}
For a round $t$ in TTC, let $\mathcal{A}' \subseteq \mathcal{A}^{\mathrm{nc}}$ 
be the set of unassigned agents.  
Each agent $a_i \in \mathcal{A}'$ points to the owner of its top remaining acceptable resource (denote this resource by $r'$), which is either (i) owned by one or more agents in $\mathcal{A}'$, or (ii) unowned.  
Accordingly, we construct a directed graph on vertex set $\mathcal{A}' \cup \{\phi\}$, where $\phi$ represents all unowned resources:  
\begin{itemize}
    \item[(i)] If $r'$ is owned by an agent $a_j$, we add a directed edge $a_i \to a_j$.  
    \item[(ii)] If $r'$ is co-owned by $|q_{r'}|$ agents, we add $|q_{r'}|$ such edges. 
    \item[(iii)] If $r'$ is unowned, we add a single edge $a_i \to \phi$.
\end{itemize}
In case (i), every agent in $\mathcal{A}'$ has out-degree at least one, so the resulting finite digraph contains a directed cycle~\cite{bang2008digraphs}.  
Each such cycle (including overlapping cycles resolved as in Case~A of Section~\ref{sec:solution_app}) is executed via a TTC rotation, which assigns resources to the agents on the cycle and removes them from the market.

In case (ii), $\phi$ has out-degree zero, and any directed walk ending at $\phi$ 
defines an agent $\phi$ chain beginning at an agent with no incoming edges. We close the chain into a directed cycle by adding an auxiliary edge $\phi \to a_{i_1}$, where $a_{i_1}$ is the first agent on the chain, and apply a TTC rotation to reallocate one unowned resource (Case~B of 
Section~\ref{sec:solution_app}).  
This again removes at least one agent and its allocated resources from the market.

Both cases may occur simultaneously, but in every round, at least one agent is 
removed from $\mathcal{A}'$.  
Finiteness of $\mathcal{A}'$ therefore implies that the procedure terminates 
after finitely many rounds, producing a feasible allocation.
\end{proof}
\vspace{-0.1in}
\begin{theorem}[\textbf{Individual rationality}]\label{thm:IR_case1_evonly}
The TTC outcome $\mu$ is individually rational:
for every agent $a_i$, $\mu(a_i)\succeq_s \omega(a_i)$.
\end{theorem}
\vspace{-0.15in}
\begin{proof}
For a fixed round $t$ and let $\mathcal{A}'$ be the set of agents who are unassigned.  Each $a_i\in\mathcal{A}'$ points to the owner of its top remaining acceptable resource $r'$, which is either (i) owned by one or more agents in $\mathcal{A}'$, or  
(ii) unowned (represented by a dummy vertex $\phi$). This defines a directed graph on $\mathcal{A}'\cup\{\phi\}$.

\noindent \emph{(i) An agent's endowed resource cannot be taken without the agent being part of the same resolved cycle.} Whenever a pointed resource is owned, an edge is drawn toward the owner of that resource.  Thus, if an endowed resource $\omega(a_i)$ is allocated to another agent during a TTC rotation, that allocation must come from a directed cycle (Case~A) or from a chain closed into a cycle (Case~B). In either case, $\omega(a_i)$ is passed only to the predecessor of $a_i$ on the cycle.  Hence, if an agent’s endowment leaves them, they must lie on that same cycle, ensuring no endowment is taken without including its owner.

\noindent \emph{(ii) When an agent exits, it receives its top remaining acceptable resource.}  
At the exit round, each agent $a_i$ points to its top remaining acceptable resource $r'$, and by definition, the endowment $\omega(a_i)$ is still among its remaining acceptable resources until $a_i$ exits. If $r'$ is owned, the TTC rotation assigns $a_i$ exactly $r'$, whether in a simple cycle, an overlapping cycle (Case~A), or a cycle formed by closing a chain (Case~B). If $r'$ is unowned, the chain is closed into a cycle via $\phi\to a_i$, and the rotation again assigns $a_i$ the unowned resource it pointed to. In the degenerate self-loop case, $r'=\omega(a_i)$, and $a_i$ simply retains its endowment.

Thus, at its exit round, every agent receives exactly its top remaining acceptable resource, which is always weakly better than its endowment. Since each agent receives a resource that is weakly preferred to its endowment at the moment of exit, we have  
$\mu(a_i)\succeq_{a_i}\omega(a_i) \, \text{for every agent} \, a_i$. Therefore, the final allocation is individually rational.
\end{proof}
\vspace{-0.15in}
\begin{theorem}[\textbf{Pareto Optimality}]\label{thm:PE_case1_evonly}
The outcome of TTC, i.e., $\mu$, is always Pareto optimal.
\end{theorem}
\vspace{-0.15in}
\begin{proof}
Let $\mu$ be the final allocation produced by the proposed algorithm, and suppose, for contradiction, that there exists a feasible allocation $\mu'$ that Pareto improves upon $\mu$, i.e., $\mu'(a_{i'})\succeq_{a_{i'}} \mu(a_{i'})$, $\forall a_{i'} \in \mathcal{A}^{\mathrm{nc}}\setminus{\{a_i}\}$, and $\exists \,\, a_i, \text{such that} \, \mu'(a_i)\succ_{a_i} \mu(a_i)$. Among all agents who are strictly better off under $\mu'$ than $\mu$, let $a_i$ be the agent who is removed \emph{earliest} by our algorithm. Let $r$ be a resource such that $r\in \mu'(a_i)$ and $r\notin \mu(a_i)$, and let $a_{i'}$ be the agent who receives $r$ in the TTC round (say $t_r$). These allocations can be made via a simple cycle, an overlapping cycle, or a chain closing into a cycle.
\vspace{0.05in}

\noindent \emph{Claim 1. The resource $r$ was still available at the round when $a_i$ 
was removed by our algorithm.}  
If we assume the opposite, then $r$ must have been allocated in an earlier round than $t_r$ to some agent $a_{i'} \neq a_i$. Since $\mu'$ assigns $r$ to $a_{i}$ instead of 
$a_{i'}$, and $\mu'$ is a Pareto improvement over $\mu$, agent $a_{i'}$ cannot be 
worse off under $\mu'$. Because preferences are strict over individual resources, if $a_{i'}$ does not receive $r$ in $\mu'$, it must receive a resource that is strictly 
better than $r$. Thus, $a_{i'}$ is also strictly better off under $\mu'$ than under $\mu$, and since $r$ is allocated at round $t_r$, agent $a_{i'}$ must exit \emph{no later} than $a_i$. This contradicts the choice of $a_i$ as the strictly improving agent who exits earliest. Hence, $r$ must have been available at the round when $a_i$ was removed.
\vspace{0.05in}

\noindent \emph{Claim 2. At the round when $a_i$ is removed, the algorithm assigns it a resource that is at least as preferred as $r$.}  At its exit round, $a_i$ points to its top remaining acceptable resource. Since $r$ is available at this round by Claim~1, we have two possibilities: (a) $a_i$ points directly to $r$. In this case, the cycle (or chain closed into a cycle) is executed in that round, assigning $r$ to $a_i$ immediately. (b) $a_i$ points to some resource $y$ with $y \succ_{a_i} r$.  
Then, in that round, the TTC rotation assigns $y$ to $a_i$, and therefore $a_i$ receives a resource strictly preferred to $r$. In all cases of the algorithm, i.e., simple cycles, overlapping cycles, and chain-cycles via the dummy vertex, the agent that was removed in that round receives exactly the resource it points to. Thus, at its exit, $a_i$ receives a resource that is $\succeq_{a_i} r$.

Combining Claims~1 and~2, we conclude that at the round in which $a_i$ is removed by our algorithm, $a_i$ receives a resource that is 
$\succeq_{a_i} r$. But $r$ is the resource that makes $a_i$ strictly better off under $\mu'$, so $\mu(a_i)\succeq_{a_i} r$ implies that $\mu(a_i)\succeq_{a_i} \mu'(a_i)$, contradicting the assumption that 
$\mu'(a_i)\succ_{a_i} \mu(a_i)$.  
Thus, no feasible allocation can strictly improve upon $\mu$ without making some agent worse off, and therefore $\mu$ is Pareto efficient.
\end{proof}
\vspace{-0.15in}
\begin{theorem}[\textbf{Core Stability}]\label{thm:core_case1_evonly}
TTC outcome $\mu$ is in the core, implying no coalition $\mathcal{A}'\subseteq \mathcal{A}^{\mathrm{nc}}$ can reassign its initially endowments so that \emph{every} member of $\mathcal{A}'$ is strictly better than under $\mu$.
\end{theorem}
\vspace{-0.15in}
\begin{proof}
Assume, for contradiction, that there exists a non-empty coalition $\mathcal{A}' \subseteq \mathcal{A}^{\mathrm{nc}}$ and a reassignment $\mu'$ that uses only the coalition’s endowed units $W'=\{\omega(a_i):a_i\in \mathcal{A}'\}$ and satisfies $\mu'(a_i)\succ_{a_i} \mu(a_i)$ for every $a_i\in \mathcal{A}'$. Let us consider an agent $a_i \in \mathcal{A}'$ who is removed \emph{earliest} by the algorithm. Let $r =\mu'(a_i)\in W'$ be the resource that
$a_i$ would receive under the blocking deviation.

\noindent \emph{(1) All coalition endowments are still available when $a_i$ exits.}
In our setting, an endowed unit can be transferred only when the owner participates in a TTC cycle. Since $a_i$ is the first member of $\mathcal{A}'$ to exit, no other agent has exited before it. Therefore, none of the units in $W'$ has been removed yet, and in particular $r$ is still available when $a_i$ is processed.
\vspace{0.05in}

\noindent \emph{(2) At its exit round, $a_i$ receives its top remaining acceptable unit.}
In that round's $a_i$ points to the owner of its top remaining choice. The algorithm, whether it executes a simple cycle, an
overlapping cycle, or a chain closed into a cycle, assigns it to $a_i$. Since $r$ is available, two cases arise: (i) $a_i$'s top choice is $r$, and is assigned, so $\mu(a_i) = r$, contradicting $\mu'(a_i) \succ_{a_i} \mu(a_i)$. (ii) $a_i$'s top choice is $y$ with $y\succ_{a_i} r$. Then $\mu(i) = y\succ_{a_i} r =\mu'(a_i)$, contradicting that $\mu'$ makes every coalition member strictly better off.

Both cases contradict the assumption that $\mu'$ strictly improves every agent in $\mathcal{A}'$. Therefore, no such blocking coalition exists, and $\mu$ is in the core.
\end{proof}

\begin{figure*}[t]
    \centering
    \begin{minipage}[t]{0.49\textwidth}
    \centering
    \begin{subfigure}[b]{0.49\textwidth}
        \includegraphics[width=\linewidth]{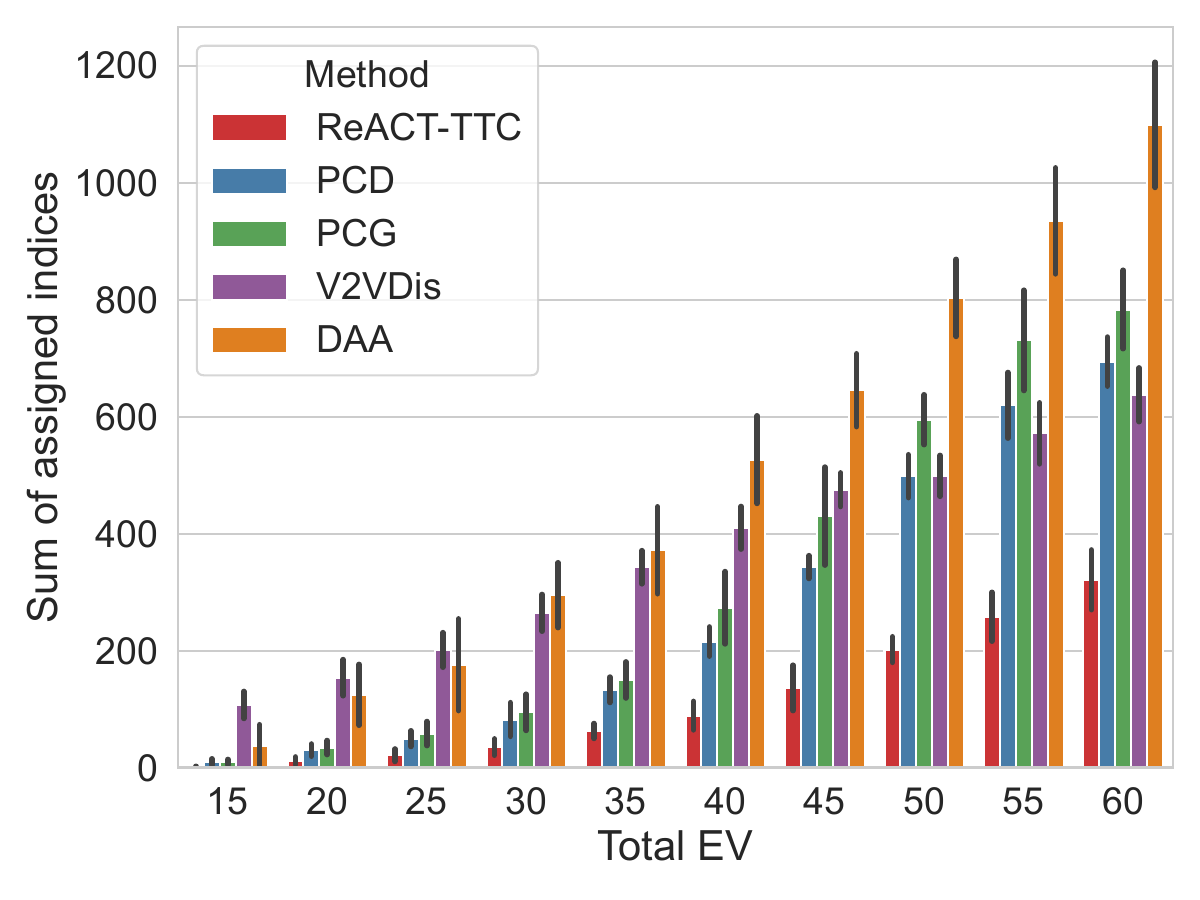}
        \caption{Sum of Assigned Indices}
        \label{fig:sum_indices}
    \end{subfigure}
    \hspace{-0.02in}
    \begin{subfigure}[b]{0.49\textwidth}
        \includegraphics[width=\linewidth]{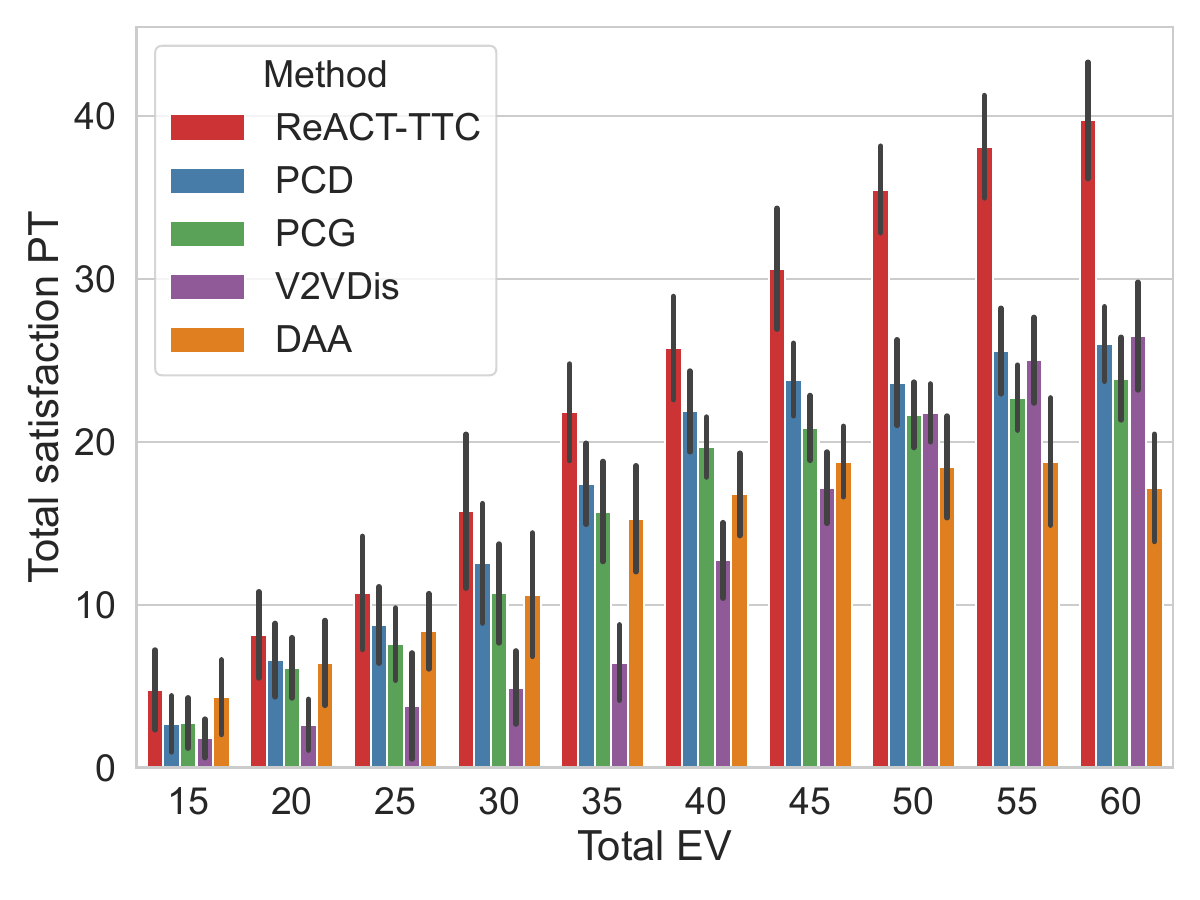}
        \caption{Satisfaction}
        \label{fig:sat_pt}
    \end{subfigure}
    \vspace{-0.1in}
    \caption{Analysis of satisfaction while varying the number of noncompliant EVs.}
     \vspace{-0.2in}
    \label{fig:varying_ev}
\end{minipage}%
  \hfill
\begin{minipage}[t]{0.49\textwidth}
    \centering
    
    \begin{subfigure}[b]{0.49\textwidth}
        \includegraphics[width=\linewidth]{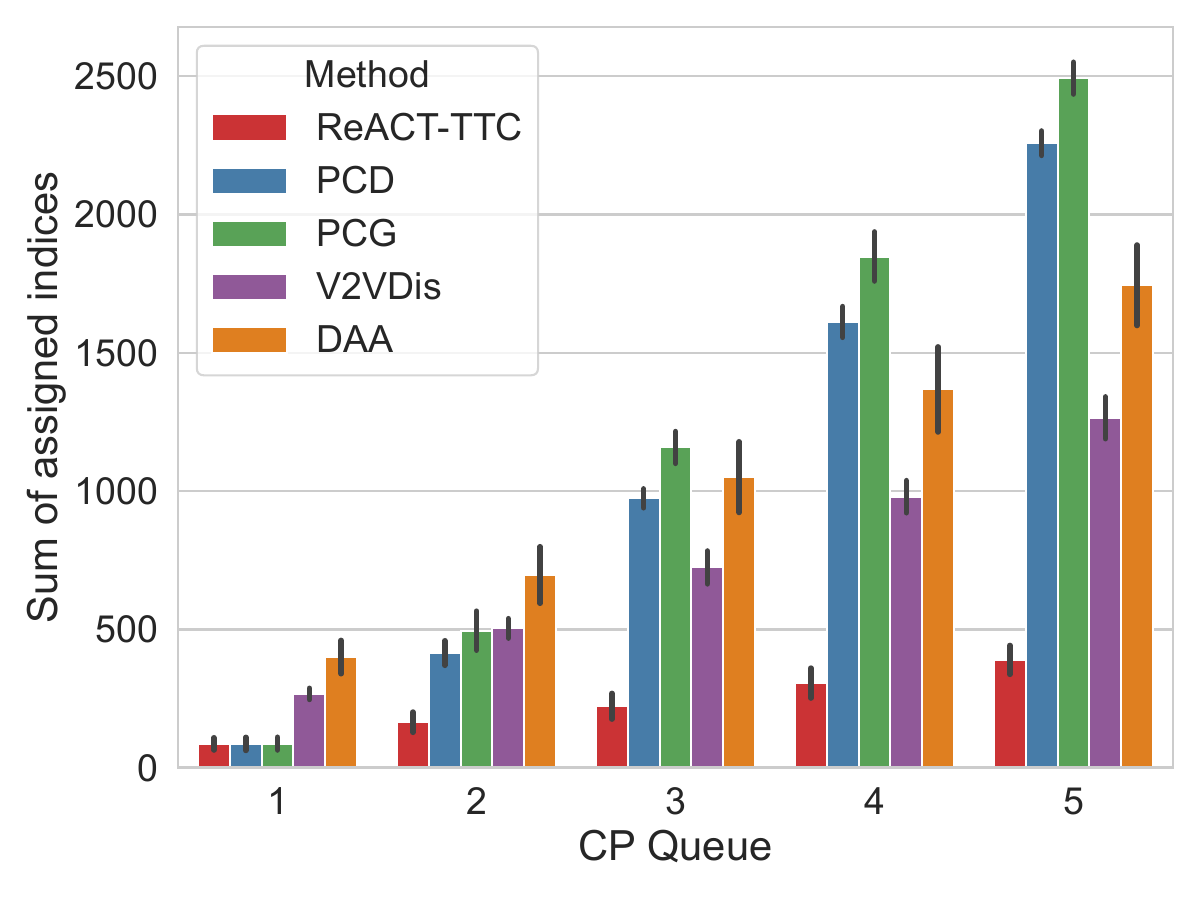}
        \caption{Sum of Assigned Indices}
        \label{fig:sum_indices_Q}
    \end{subfigure}
    \begin{subfigure}[b]{0.49\textwidth}
        \includegraphics[width=\linewidth]{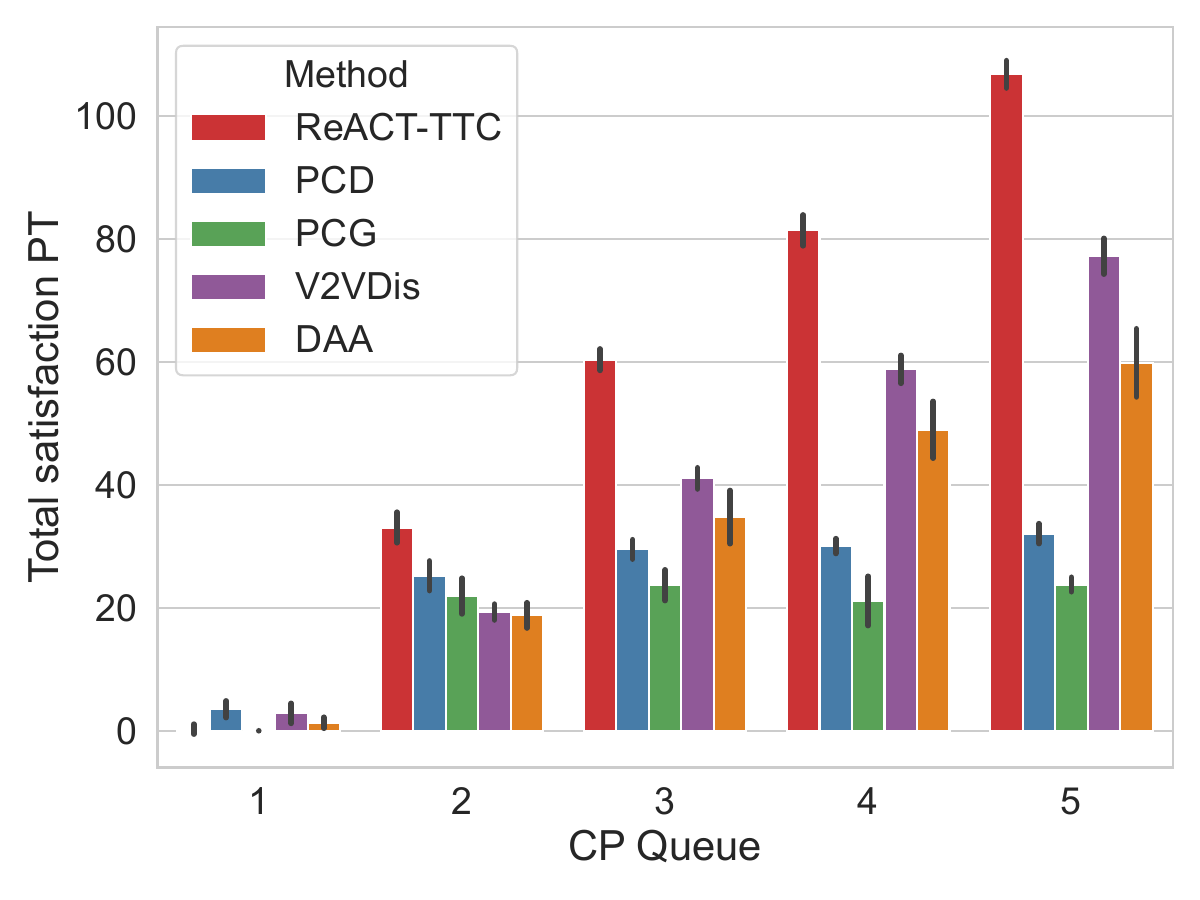}
        \caption{Satisfaction}
        \label{fig:sat_pt_Q}
    \end{subfigure}
    \vspace{-0.1in}
    \caption{Analysis of satisfaction while varying the size of the queue at a charging point. The non-compliant EV-to-total charging slot (Total CP * queue) ratio is fixed at $0.8$.}
    \vspace{-0.2in}
    \label{fig:varying_Q}
\end{minipage}
\end{figure*}
\vspace{-0.15in}
\begin{theorem}[\textbf{Strategy-Proofness}]\label{thm:SP_case1_evonly}
The proposed algorithm is individually
strategy-proof under strict preferences. For any agent $a_i$, holding all other
reports fixed, truthful reporting yields an allocation that is at least as good
for $a_i$ (in its true preference order) as the allocation obtained under any
misreport.
\end{theorem}
\vspace{-0.15in}
\begin{proof}
Let us fix the preference of all agents other than $a_i$. Then we consider two execution mechanisms \emph{Truthful} (T), where $a_i$ reports its true preference $\succ_{a_i}$, and \emph{Misreport} (M), where its reported preference differs. For contradiction, we assume that in M, agent $a_i$ receives a resource that $a_i$ strictly prefers to the one received in T. Let $t_r$ be the earliest round at which the two executions differ in the set of assignments made. By definition of $t_r$, the sets of remaining agents and resources at the start of $t_r$ are identical in T and M, and so is the ownership. 

\noindent \emph{(1) Every directed cycle in round $t_r$ contains $a_i$.} Under TTC implementation, if there existed a cycle that does not contain $a_i$, then this cycle would appear identically in both T and M, and the mechanism would execute it in both runs. This would imply that T and M make the same
assignments in round $t_r$, contradicting the definition of $t_r$ as the first round where the assignments differ. Hence, every cycle present in round $t_r$ must contain $a_i$.

\noindent \emph{(2) In T, $a_i$ receives its top remaining acceptable resource at
round $t_r$.} Let $r^\star$ denote $a_i$'s top remaining acceptable choice at the start
of $t_r$ according to its true preference order $\succ_{a_i}$. In the truthful
execution T, agent $a_i$ points to the owner of $r^\star$. By (1), any cycle
executed in round $t_r$ must contain $a_i$, and TTC assigns to $a_i$ exactly what it points to on that cycle. Therefore, in T, $a_i$ receives $r^\star$ at round $t_r$.

\noindent \emph{(3) $a_i$ cannot strictly improve in M.}
At round $t_r$, the set of remaining resources is the same in T and M,
so any resource that $a_i$ can receive in M at round $t_r$ is among the same remaining
units considered in T. By definition of $r^\star$, we have $r^\star \succ_{a_i} r$ or $r^\star = r$ for every other remaining resources $r$. Thus, whatever unit $a_i$ receives in M at round $t_r$ is weakly worse than $r^\star$ in $\succ_{a_i}$. Since in T, $a_i$ receives $r^\star$, agent $a_i$ cannot
obtain a resource that is strictly preferred to its T-assignment by misreporting,
contradicting the assumption.

Hence, no misreport by $a_i$ can make $a_i$ strictly better off, and the mechanism
is individually strategy-proof.
\end{proof}

\vspace{-0.15in}
\section{Experimental Results}\label{sec:results_new}
This section presents the experimental environment used to validate our approach in a realistic intelligent transportation system. We model an operational electric vehicle (EV)-charging CPS in which EVs interact with charging points (CPs) under heterogeneous preferences, varying compliance behavior due to dynamically evolving system states. The system is viewed as an assignment environment, where EVs constitute set of agents and CPs form the resources, with each CP with a quota $q_j \geq 1, \forall r_j \in \mathcal{R}$. Consistent with existing literature, several charge recommendation algorithms generate initial EV--CP assignments. However, these recommendations are typically system-centric and based on static assumptions, which may lead to non-compliance from human drivers. When an EV user chooses not to comply with the initial recommendation, our proposed \textsc{ReACT-TTC} mechanism incorporates the user's current preferences and computes an improved, preference-aligned reassignment without impacting the assignment of the complaint users.

\vspace{-0.15in}
\subsection{Environmental Setup}
All simulations were implemented in \texttt{Python 3.8.3}, with EVs and CPs represented as independent object classes. The implementation of the \textsc{ReACT-TTC} and the baselines is available at~\cite{IoTLab02_react_TTC_2025}.

\noindent \textbf{Dataset Description:}
To evaluate \textsc{ReACT-TTC} under realistic operating conditions, we use a benchmark derived from the real-world \texttt{JD Logistics} distribution network in Beijing \cite{zheng2024hybrid}. The dataset contains actual delivery and pickup requests, from which we sample locations to construct problem instances of varying sizes, 
each paired with a fixed set of charging stations. 
The dataset provides the geographic coordinates of all EV and CP nodes, along with pairwise travel distances and travel times that may violate the triangle inequality, reflecting real-world conditions of urban mobility. EV energy consumption parameters and charging characteristics are taken directly from the dataset's operational estimates. As a result, the generated scenarios capture realistic spatial structures, demand patterns, and travel behaviors, which are suitable for evaluating EV-CP assignment and reassignment mechanisms.

\noindent\textbf{Baseline Algorithms:} 
To evaluate the effectiveness of \textsc{ReACT-TTC}, we compare it against three state-of-the-art matching-based approaches. 
\textsc{V2VDisCS}~\cite{chatterjee2025v2vdiscs} formulates V2V charge sharing as a stable matching problem on bipartite graphs, solved via integer linear programming with low-overhead distributed heuristics in limited-communication settings. 
\textsc{DAA}~\cite{akter2024move} assigns EVs to roadside resources via a stable matching framework. For fairness, we use the same EV preference model and generate CP preferences by sorting EVs based on their maximum charging needs. 
\textsc{SMEVCA}~\cite{khanda2025smevca} models EV--CP assignment as a one-to-many matching game under a subscription-driven SLA framework, using roadside units to form coalitions through two strategies: a fast greedy heuristic (PCG) and a more computation-intensive dynamic optimizer (PCD). For fair comparison, we adopt the preference generation procedure of \textsc{SMEVCA} and the initial endowment is generated using a local search-based algorithm with the same objectives as in SMEVCA.
\vspace{-0.15in}
\subsection{Results and Analysis}
Fig.~\ref{fig:varying_ev} illustrates the performance of different schemes with increasing non-compliant users. 
For this experiment, we assume the CP capacity $q_j= 2, \forall r_j \in \mathcal{R}$, hence the maximum limit (capacity constraint of the CPS system) is $60$ EVs at full capacity. As shown in Fig.~\ref{fig:sum_indices}, ReACT-TTC achieves the lowest aggregate rank sum, indicating consistently better allocations for non-compliant agents. This advantage stems from two key factors. \textit{First}, TTC ensures Pareto-optimality by constructing directed graphs from each agent’s most preferred choice, enabling mutually beneficial trades. \textit{Second}, resolving overlapping cycles and chains prioritizes minimal satisfaction decay, consistently yielding lower-rank assignments. In contrast, baseline methods such as PCD, PCG, DAA, and V2VDisCS optimize system-centric metrics (e.g., detour distance, charging cost, or communication overhead) with limited regard for user satisfaction. Between these, PCG and PCD perform relatively better since their preference models align with system objectives, whereas DAA and V2VDisCS lack such alignment, resulting in inferior satisfaction scores. Fig.~\ref{fig:sat_pt} shows that ReACT-TTC consistently attains the highest PT-satisfaction, driven by factors discussed earlier.

Fig.~\ref{fig:varying_Q} captures the aggregate ranks and the PT-satisfaction for different schemes for increasing queue size at the CPs, $q_j, \, \forall r_j \in \mathcal{R}$ for a constant ratio of non-compliance to total charging slot (= Total CP $*$ queue), set to $0.8$. In this case, referring to Fig.~\ref{fig:sum_indices_Q}, we observe an interesting behavior, where $q_j = 1, \forall r_j \in \mathcal{R}$. The sum of the satisfaction scores is at par with PCD and PCG, whereas the PT-satisfaction (Fig.~\ref{fig:sat_pt_Q}) observes an opposite behavior where PCD, PCG, and DAA attain better satisfaction. This is owing to the following two reasons. \textit{Firstly}, ReACT-TTC takes Pareto optimal decisions to reduce the rank degradation, which may not necessarily align with the goal of maximizing the total satisfaction (Eq~\eqref{obj:max}) of the agents. \textit{Secondly}, the gain in ReACT-TTC is computed with reference to the initial endowment, whereas for the others, any assignment beyond the reference point contributes to $0$ satisfaction, reducing the negative impact beyond the endowments. Additionally, we observe that beyond $q_j \geq 2$, ReACT-TTC outperforms all baselines. As the queue size increases, the trade options for non-compliant agents also increase, thereby boosting the chances of better assignments.
\begin{figure}[t]
    \centering
    \begin{subfigure}[b]{0.49\linewidth}
        \includegraphics[width=\linewidth]{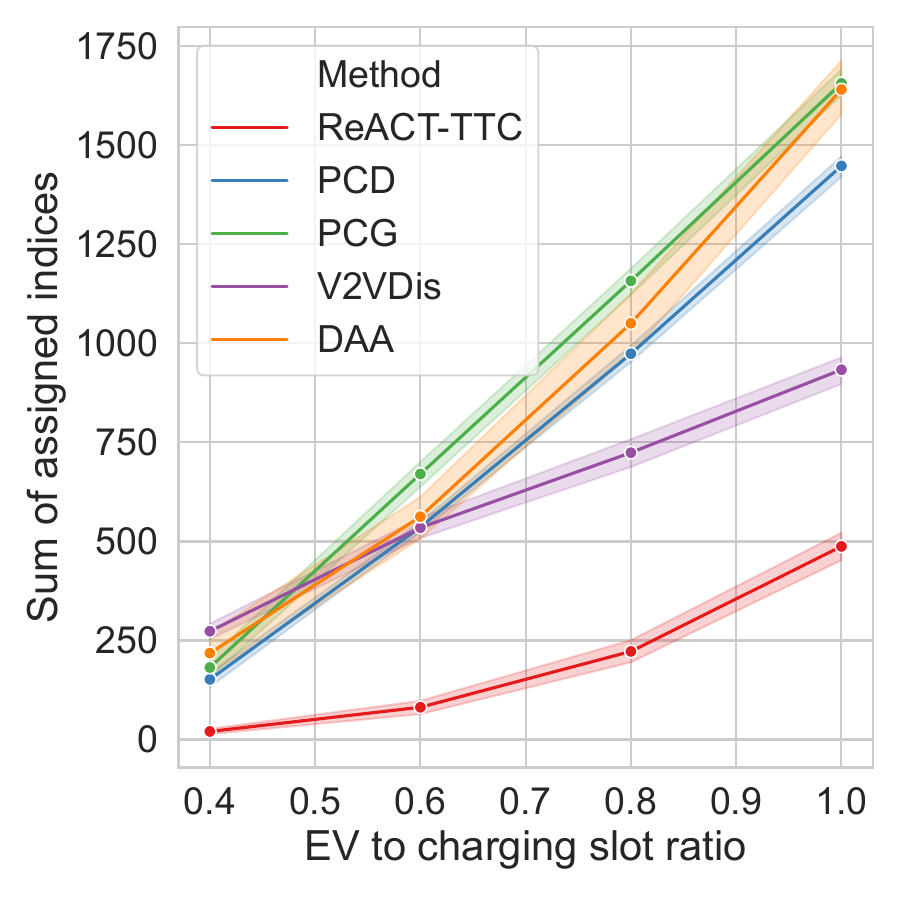}
        \caption{Sum of Assigned Indices}
        \label{fig:sum_indices_ratio}
    \end{subfigure}
    \hfill
    \begin{subfigure}[b]{0.49\linewidth}
        \includegraphics[width=\linewidth]{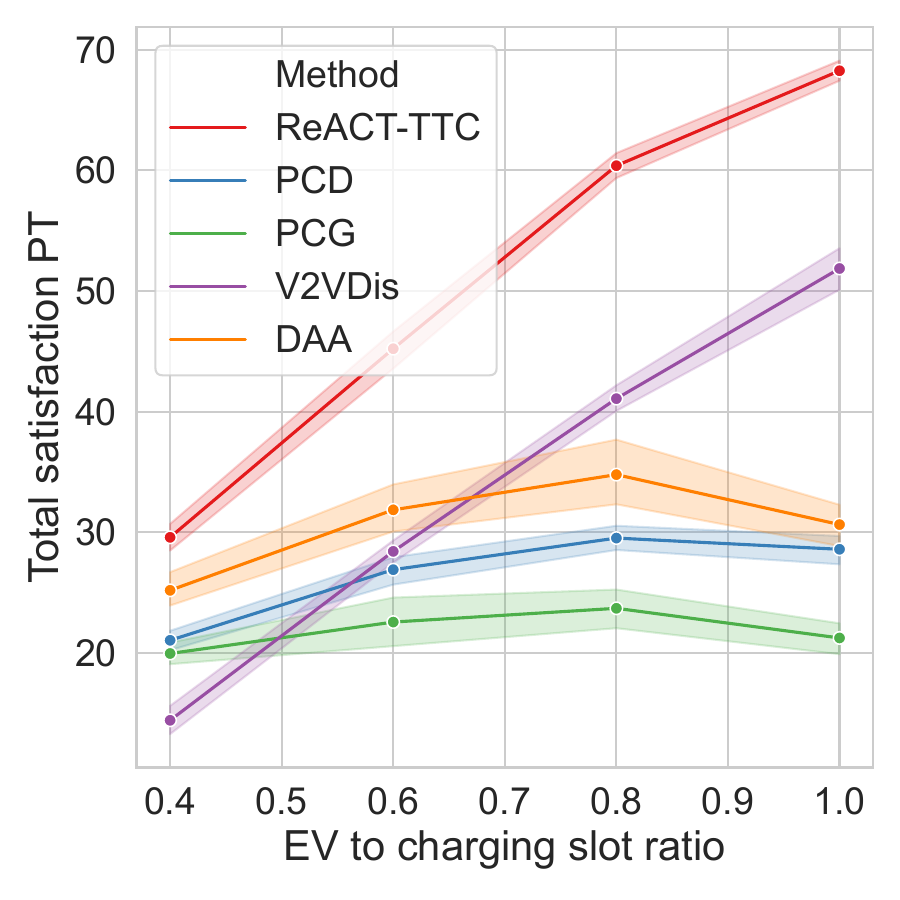}
        \caption{Satisfaction}
        \label{fig:sat_pt_ratio}
    \end{subfigure}
    \vspace{-0.1in}
    \caption{Analysis of satisfaction while varying the Non-compliant EV to total charging slot, i.e. $q_j = 3, \forall r_j \in \mathcal{R}$.}
    \vspace{-0.2in}
    \label{fig:varying_ratio}
\end{figure}
For the next set of experiments, we fix the queue size to $q_j = 3, \forall r_j \in \mathcal{R}$ and vary the ratio of non-compliance to the total number of charging slots. Across all schemes, the overall sum increases as the number of non-compliant users increases, as expected. It can also be observed from Fig.~\ref{fig:varying_ratio} that ReACT-TTC outperforms the baselines in both metrics (refer to Figs. \ref{fig:sum_indices_ratio} and \ref{fig:sat_pt_ratio}) considered earlier, for reasons explained earlier.

\vspace{-0.1in}
\section{Discussions}\label{sec:discussion}
ReACT-TTC provides an efficient reassignment mechanism, tested on EV–CP allocations under multi-quota constraints. While it maintains high satisfaction under non-compliance, some extensions can enhance its performance.
\begin{itemize}
    \item \textbf{Cycle Sequencing Sensitivity}: The sequence of resolving overlapping cycles critically impacts system satisfaction. Dependencies and competing improvements across shared edges create a complex combinatorial space, making optimal ordering non-trivial. Developing principled sequencing strategies still remains an open challenge.
    
    \item \textbf{Parallel Cycle Detection}: Current cycle detection via depth-first search (DFS) is sequential and limits scalability in dense graphs. Parallelization could accelerate reassignment but may introduce conflicts across shared agents and edges. Safe coordination for parallel detection is a promising direction.
    
    \item \textbf{Learning-Based Preferences}: Although preferences are randomly generated for domain independence, real CPS systems provide data that can inform preference learning and better endowments. Incorporating such models can more accurately capture agent behavior and improve satisfaction.
    
   \item \textbf{Cross-Domain Evaluation}: Beyond the EV--CP setting, ReACT-TTC naturally extends to CPS domains where agents compete for limited resources and reassignment can improve system performance. In compute-task allocation, latency-sensitive tasks act as agents, and edge servers serve as resources; exchanging task-server assignments enables load consolidation and reduces queueing delay. In UAV bandwidth sharing, UAVs are agents, and communication channels are scarce resources; reallocating channel assignments can help mitigate interference hotspots and improve link stability. Such cross-domain evaluations can reveal domain-specific gains arising from reassignment.
\end{itemize}

\vspace{-0.1in}
\section{Conclusions}\label{sec:cnls_new}
This work presents a post-deviation reassignment framework tailored for shared-resource CPS in which user non-compliance is unavoidable. By extending the TTC mechanism to support many-to-one capacities and unassigned resources, we enable preference-driven reallocation that preserves core guarantees, including Pareto efficiency, individual rationality, and strategy-proofness. Our capacity-aware cycle-detection rules ensure correctness and termination, while integrating Prospect-Theoretic preferences provides a more realistic model of user satisfaction. The framework operates independently of the initial allocation method, incurs minimal re-computations, and responds only when deviation occurs, making it suitable for real-time CPS environments. An evaluation of a real-world EV-charging scenario shows that the framework improves user satisfaction and assignment quality under heterogeneous behavior, highlighting its potential as a lightweight, domain-agnostic mechanism for resilient human-in-the-loop CPS.

\bibliographystyle{ACM-Reference-Format}
\bibliography{refs}

\end{document}